\documentclass[runningheads]{llncs}

\usepackage[T1]{fontenc}
\usepackage{amsmath,amssymb,mathtools}
\usepackage{graphicx}
\usepackage{booktabs,array}
\usepackage{enumitem}
\usepackage{xcolor}
\usepackage{aliascnt}
\usepackage{microtype}
\usepackage{tikz}
\usetikzlibrary{arrows.meta}
\usetikzlibrary{decorations.pathreplacing}

\usepackage[colorlinks=true,linkcolor=blue!60!black,citecolor=blue!60!black,
            urlcolor=blue!60!black]{hyperref}
\usepackage[capitalize,noabbrev]{cleveref}

\makeatletter

\let\c@theorem\relax   
\let\c@lemma\relax     
\makeatother
\usepackage{amsthm}

\newtheorem{theorem}{Theorem}[section]
\newaliascnt{lemma}{theorem}
\newtheorem{lemma}[lemma]{Lemma}
\aliascntresetthe{lemma}
\crefname{theorem}{Theorem}{Theorems}
\crefname{lemma}{Lemma}{Lemmas}

\newlength{\evcolL}
\newsavebox{\frlbox}\newsavebox{\frrbox}
\newlength{\frht}\newlength{\frtmp}
\newcommand{\frmeasure}{%
  \setlength{\frht}{\ht\frlbox}\addtolength{\frht}{\dp\frlbox}%
  \setlength{\frtmp}{\ht\frrbox}\addtolength{\frtmp}{\dp\frrbox}%
  \ifdim\frtmp>\frht\setlength{\frht}{\frtmp}\fi}
\newcommand{\frslot}[1]{\parbox[c][\frht][c]{\linewidth}{\centering#1}}

\newcommand{\Mod}{\mathsf{Mod}}
\newcommand{\val}{\mathsf{val}}
\newcommand{\prev}{\mathsf{prev}}
\newcommand{\last}{\mathsf{last}}
\newcommand{\enc}{\mathsf{enc}}
\newcommand{\MR}{\mathsf{MR}}
\newcommand{\Prove}{\mathsf{Prove}}
\newcommand{\Verify}{\mathsf{Verify}}
\newcommand{\keys}{\mathsf{keys}}

\newcommand{\tomb}{\bot}
\newcommand{\absent}{\varnothing}
\newcommand{\window}[2]{W_{#1,#2}}
\newcommand{\out}{\mathsf{out}}

\title{A BUD You Can Count On: Efficient, Verifiable Membership Proofs}
\titlerunning{A BUD You Can Count On}

\author{Alejandro Ranchal-Pedrosa\inst{1} \and
Cody Littley\inst{1} \and Yiming Zang\inst{1} \and
Ben Marsh\inst{1,2}}

\authorrunning{A. Ranchal-Pedrosa et al.}

\institute{Sei Labs \and University of Portsmouth}

\institute{}

\begin{document}
\maketitle

\begin{abstract}
A block typically changes only a small fraction of blockchain state, yet
producing a state root makes validators update an authenticated tree spanning
it all.
We show how to support verifiable reads from the writes alone. We introduce
\emph{Block Update Digests} (BUDs), that commit to a block's writes and link each
key touched in that block to its previous modification. An optional prior-value
record lets one successor BUD authenticate the value between modifications,
replacing the baseline's two BUD openings. \emph{SuperBUDs} summarize those writes
over longer intervals, letting a reader verify that a value remained
unchanged without checking every intervening block. Thanks to BUDs and SuperBUDs, validators can just keep ordinary
flat state while still serving proofs of current and historical values,
including deletion. We prove soundness and completeness for reads at or after a key's first
logged modification when the required proof data is available. In experiments
with the baseline format, a $50\times$
increase in state size raises the measured base-BUD cost by just $1.24\times$.
In paired QMDB experiments, lowering its garbage-collection activation
threshold raises process CPU time by $1.3$--$2.7\times$ and device bytes
written by $1.6$--$2.4\times$ relative to its dormant configuration.
BUD authentication requires no corresponding relocation of live entries.
In the evaluated baseline, two-digest
read-layer payloads stay below $800$ bytes, and warm hash-path verification takes at most
$146\,\mu\mathrm{s}$ at p99. Authenticated reads need not make validators
maintain an authenticated tree over the entire state.
\end{abstract}

\section{Introduction}
Blockchain validators execute an agreed sequence of blocks to update a
shared state. Each block typically changes only a small part of it, yet
producing the state root updates an authenticated structure spanning every
account and storage key. A signed root and a path make reads easy to prove,
but every validator pays for those proofs in the execution hot-path: each of $w$
writes changes an $O(\log N)$ path through $N$ state entries. In a standard
EVM client, state-root computation can cost roughly an order of magnitude
more than transaction execution~\cite{megaeth}.

Existing designs make this global structure cheaper, compute its root
later, or move it to specialized nodes. QMDB combines authenticated storage
with historical reads~\cite{qmdb}; Guppy moves the global tree off-chain
and proves its updates recursively~\cite{guppy}. We show that historical
point reads do not require a global tree.

Serving membership proofs requires the latest written value to the key, and a proof that no later write
changed it. A \emph{Block Update Digest} (BUD) commits to one block's writes
and links each key to its preceding modification. \emph{SuperBUDs}
summarize longer intervals, ruling out intervening writes without having to check
each block. Using BUDs, validators must thus just keep flat state.

We describe a protocol family with two BUD record formats. The baseline
stores the new value and a predecessor link. The prior-value variant also
stores the value before the block's first write, letting a single successor
record supply both the earlier value and evidence that it remained unchanged.
This removes one BUD opening and its separate attestation from a successor
certificate, at the cost of larger records. It does not require signature
aggregation, and it does not remove the need for SuperBUDs when no suitable
successor is available.

The benefit is largest for sparse updates. For $100$ writes among $10^5$
keys, base-BUD replay takes $0.77\,\mu\mathrm{s}$ per write, versus
$7.26\,\mu\mathrm{s}$ for the in-memory Merkle Patricia trie: a
$9.4\times$ difference. Across a $50\times$ increase in state size,
base-BUD cost grows by $1.24\times$. Database access still depends on
storage and caching, but authentication no longer traverses the full state.

QMDB shows why fast commitments alone are insufficient. Its garbage
collector relocates live authenticated entries on the update path.
Activating that code at the tested state sizes raises CPU time by
$1.3$--$2.7\times$ and device bytes written by $1.6$--$2.4\times$, even
where median latency improves (\cref{tab:eval:q7}). BUDs decouple
authentication from this relocation: local storage maintenance does not
rebuild historical commitments.

\paragraph{Contributions.} 
We construct historical membership and exclusion proofs from per-block
BUDs with predecessor records and mergeable SuperBUD summaries
(\cref{sec:model,sec:construction}); characterize the trade-off between
summary maintenance, proof size, and waiting time
(\cref{sec:construction,sec:frontier} and Appendix~\ref{app:schedules}); and prove soundness against
Byzantine validators and adversarial provers, with completeness for reads
anchored after deployment under the stated availability and committee
assumptions, with full proofs in Appendix~\ref{app:correctness}. Our implementation compares the
base-BUD path with Merkle Patricia tries, NOMT, and QMDB, and evaluates
hierarchy and proof costs on synthetic workloads and an Ethereum
account-access trace (\cref{sec:eval}). We also define prior-value records
and single-BUD successor certificates, with corresponding correctness and
byte-cost arguments (\cref{sec:prior-value,sec:prior-cert,sec:costs}).

\section{Related Work}
\label{sec:related}

The main distinction is whether a system authenticates the full state or
the writes needed to prove a particular read.

QMDB authenticates current and historical key versions~\cite{qmdb}.
Its back-references make it a close comparison, but entries remain under
a global authenticated structure. Garbage collection relocates live
entries, updates their index and authentication, and runs in threads on
which block completion waits. BUDs authenticate immutable per-block logs;
collecting local tombstones does not rewrite other authenticated entries.
Our paired experiments separate relocation work from implementation-specific
stalls (\cref{sec:eval:q7}). NOMT and LVMT also make global authenticated
storage cheaper~\cite{nomt,lvmt,mlsm}.

Monad and Ethereum defer a global root, while MegaETH delegates its
maintenance~\cite{monadasync,eip7862,megaeth}. Guppy~\cite{guppy} has
validators commit only to updates, but an off-chain service maintains a
Merkle tree of stream states and recursively proves its updates. Reads
combine a reusable checkpoint proof with a Merkle path; Guppy-B
parallelizes proving while retaining a sequential recursive step between
bundles. BUD archives instead assemble Merkle paths from attested writes
and summaries, without a global root or recursive proving pipeline.
The trade-off is validator-maintained summaries and proof sizes that depend
on key history and window placement. Guppy's global index gives read sizes
independent of time since the last write and supports derived indexes such
as event histories. Its present protocol leaves non-inclusion to future
work; BUDs prove absence through deletions and bounded sentinel intervals. It is an interesting direction combining Guppy's succint proofs for old state with BUDs and SuperBUDs practicality for frequently modified one.

Solana and Sui use incremental multiset digests whose update cost is
independent of state size, but which do not themselves provide point proofs
\cite{simd0215,simd0223,suicheckpoints,bellare1997,lewi2019}.
Ethereum's post-execution access list retains the global root and requires
revealing the list to prove one entry~\cite{eip7928}. Stateless-client and
vector-commitment designs retain a global dictionary commitment
\cite{utreexo,bbf,edrax,asvc,pointproofs,hyperproofs,aardvark}, with a
fundamental cost for keeping witnesses current~\cite{christ2023}.
Authenticated logs and persistent dictionaries commit snapshots or
cumulative history~\cite{coniks,crosby2009,pads2001,balloon}. BUDs instead
combine predecessor links and window summaries across independent block
commitments. Further comparisons appear in \cref{app:related}.

\section{Model and preliminaries}
\label{sec:model}

\paragraph{Blockchain, state, and write logs.}
Consensus finalizes a common sequence of blocks at heights $n=1,2,\ldots$,
which honest validators execute deterministically against a shared key-value
state. The canonical log $\Delta_n=\{(k_1,v_1),\ldots,(k_{w_n},v_{w_n})\}$
retains only each key's final write in the block. Keys lie in $\mathcal K$
and values in $\mathcal V\cup\{\tomb\}$, where $\tomb$ records deletion.
Every honest validator derives the same log. Write $\keys(\Delta_n)$ for
its keys and $w_n=|\Delta_n|$.

Logging begins at $H_0\ge1$, from canonical state $\sigma_0$. Let
$\Mod(k)=\{n\ge H_0:k\in\keys(\Delta_n)\}$, and write $\out(v)=v$ for
ordinary values and $\out(\tomb)=\absent$ for absence. For $h\ge H_0$,
let $a=\max(\Mod(k)\cap[H_0,h])$ when this set is nonempty. The read result is
\begin{equation}
\label{eq:val}
\val(k,h)=
\begin{cases}
\out(v), & \text{if $a$ exists and $(k,v)\in\Delta_a$,}\\
\sigma_0(k), & \text{if $a$ does not exist and $k\in\operatorname{dom}(\sigma_0)$,}\\
\absent, & \text{otherwise.}
\end{cases}
\end{equation}

\paragraph{Adversary, network, and cryptography.}
The adversary controls at most $f$ validators in each committee that signs
digests. It can read their private state and make them behave arbitrarily.
For committee epoch $\rho$, write $\Pi_\rho$ for the committee and
$n_\rho=|\Pi_\rho|$. Our reference configuration uses $n_\rho=3f+1$ and
$\tau=2f+1$ signatures per digest. The read protocol needs only
$f<\tau\le n_\rho-f$: every accepted digest has an honest signer, and honest
validators can produce an attestation without Byzantine participation.
These conditions apply to digests of already-finalized data; consensus must
still satisfy its own assumptions (\cref{sec:attest}).

Messages may be delayed arbitrarily; neither soundness nor completeness
requires a bound on their delay. Completeness does require the proof data
to become available. The adversary runs in probabilistic polynomial time.
We assume collision-resistant hashing and signatures secure against
chosen-message forgery across multiple signers. Aggregated signatures must
also resist rogue-key attacks, for example through proof of possession at
key registration; individual signatures suffice as well. All objects use a
fixed, injective, prefix-free encoding $\enc$, including integer widths and
signedness, key and value lengths, and the deletion marker.

\paragraph{Authenticated maps.}
We authenticate a map by sorting its entries by key and building a Merkle
tree over that list. An inclusion proof opens the queried entry. An
exclusion proof shows where the key would have appeared: between two
consecutive entries, before the first, or after the last. An empty map has
a distinguished root. The proof includes leaf positions and the number of
entries, so two nonadjacent keys cannot be used to hide an entry between
them. We write $\MR_{T,I}(\mathbf z)$ for the root of the sorted list
$\mathbf z$; it also binds the object type $T$, its block height or interval
$I$, and the tree's size and shape. The exact encoding appears in
\cref{sec:maps}.

\paragraph{Problem statement.}
A membership claim says $\val(k,h)=v\in\mathcal V$; an exclusion claim says
$\val(k,h)=\absent$. A \emph{read system} lets a reader verify either claim
without holding the state or trusting the server. An untrusted prover uses
the archived logs and attestations to run $\Prove(k,h)$, returning a claimed
result $y\in\mathcal V\cup\{\absent\}$ and a certificate $\pi$. The reader
runs $\Verify(k,h,y,\pi)$ using authenticated committee information.

We require two properties. \emph{Soundness}: the probability that an
adversary makes the verifier accept $y\ne\val(k,h)$ is negligible.
\emph{Completeness}: for any query whose key has a recorded modification
between $H_0$ and $h$, an honest prover can produce a correct, accepting
certificate once the required records, maps, attestations, and committee
evidence are retrievable and every used epoch can be authenticated.
The prior-value variant also permits some reads before a key's first logged
modification when a suitable later record becomes available
(\cref{sec:prior-cert}); the completeness guarantee remains the same.

\section{The Protocol: BUDs and SuperBUDs}
\label{sec:construction}

Validators keep flat state and authenticate each block's write log with a
BUD. SuperBUDs summarize longer intervals, so proving that an old value is
still current need not check every intervening block.
\Cref{fig:superbud} illustrates both commitments;
\cref{tab:terms} collects the notation.

\subsection{BUDs}
\label{sec:metadata}

A \emph{Block Update Digest} (BUD) is a Merkle commitment to the writes of a
single finalized block. Each record contains a key, its new value, the block height, and a
link to the key's previous recorded modification. That last field lets us
connect writes across blocks and prove how long a value remained unchanged.
We first specify this baseline format, then extend it with a prior value.
We call a BUD record a \emph{budlet}.

\paragraph{Records and predecessor links.}
Each entry in the validator's flat database stores $(v,\last(k))$, where
$\last(k)$ is the height of its latest logged write. As we consider systems already in production that switch to BUDs, we prescribe entries inherited at deployment start with
$\last(k)=H_0-1$. This boundary value does not name a BUD and, by itself, cannot prove
anything about the key's earlier value.

For each $(k,v)\in\Delta_n$, execution copies the entry's existing
$\last(k)$ into $p$, or sets $p=-1$ if no entry exists. It emits
$(k,v,n,p)$ and stores $(v,n)$ at $k$. Thus $p\ge H_0$ identifies a previous
logged modification; $p=H_0-1$ identifies the first logged write to an
inherited key; and $p=-1$ means no local entry existed before the write.
Repeated writes within one block produce just one record, containing the
final value and the predecessor from before the block's first write.
The extra per-entry field occupies $8$ bytes in our implementation.

\paragraph{Deletions and tombstones.}
Deleting a key stores $(\tomb,n)$ as a \emph{tombstone}: metadata recording
absence and its deletion height. A later write can then point back to the
deletion, whose archived record supports an exclusion proof.

\paragraph{Commitment and proofs.}
Sort the records by key, writing $r_i=(k_i,v_i,n,p_i)$. The BUD is
\[
U_n=\MR_{\mathsf{BUD},n}(r_1,\ldots,r_{w_n}).
\]
Validators attest $U_n$ after block $n$ is final (\cref{sec:attest}). An
inclusion proof opens a record and checks its key and height. An exclusion
proof shows that the block did not write the key, using the sorted-map
proofs from \cref{sec:model}. Exclusion from a BUD means \emph{no write in
that block}; by itself, it does not mean the key was absent from the state.

Suppose a record at height $a$ writes $v$, and a later record at height $b$
points back to $a\ge H_0$. The pointer proves that no block in $(a,b)$
modified the key. Together, the two records establish the value throughout
$[a,b)$. For example, a write at block $10$ and its successor at block $40$
can prove the value at block $25$: the first supplies the value, and the
second proves it had not changed. If the record at $10$ is a deletion, the
same argument proves absence.

\paragraph{Garbage collection.}
Tombstones need not stay in the local database forever. A tombstone written
at height $m$ is retained until $m+\eta$, where $\eta$ is the configured
retention window; live values are never collected. Collection runs
deterministically after a block's writes and never removes an entry written
in that block.

Once a tombstone has been collected, a later write at height $b$ uses
$p=-1$. This still proves something useful: the key was absent throughout
$[b-\eta,b)\cap[H_0,b)$. Either it was absent from the start, or its last
deletion was old enough for collection. We call such a record a
\emph{sentinel}. It does not prove that the key never existed. Garbage
collection removes only local metadata; archived modification records remain
available for historical proofs. The retention window limits the reach of
sentinel proofs, not reads anchored at those archived records. Setting $\eta=\infty$ disables collection and lets a sentinel reach back
to $H_0$.

\subsubsection{Prior-value records.}
\label{sec:prior-value}

The prior-value variant emits
\[
r_i^+=(k_i,v_i^-,v_i,n,p_i),\qquad
U_n^+=\MR_{\mathsf{BUD}^{+},n}(r_1^+,\ldots,r_{w_n}^+).
\]
Here $v_i^-$ is the value held immediately before block $n$'s first write
to $k_i$, using $\tomb$ if the key is absent. The current value $v_i$ and
predecessor $p_i$ follow the baseline rules. Multiple writes in one block
preserve this original prior value and predecessor while replacing only
the final current value. Thus the prior value is not an intermediate value
from within the block. Validators capture it before overwriting state;
they need no additional persistent field in each live-state entry.

The distinct object type $\mathsf{BUD}^{+}$ and signed protocol version
bind the record format. Honest validators derive and check the prior value
from canonical execution before attesting the root. A prover cannot append
a prior value to an opening under a baseline root. A deployment selects a
format; it need not construct or sign both roots for each block.
SuperBUD maps, their merge operation, and tombstone collection are unchanged.
All baseline certificate forms remain valid with openings of the enriched
records, using their current-value and predecessor fields.

For a record $(k,v^-,v,b,p)$ with $p\ge H_0$, one opening proves
$\val(k,h)=\out(v^-)$ for $p\le h<b$, and
$\val(k,b)=\out(v)$. The two values therefore answer reads throughout
$[p,b]$, with the current value used at the right endpoint. The single-BUD
verification rules, including inherited entries and collected tombstones,
appear in \cref{sec:prior-cert}.

\subsubsection{Touch transactions.}
\label{sec:touch}

A reader can request a successor link through a \emph{touch transaction}:
a write preserving the current value.
Once finalized and attested at $t$, its record $(k,v,t,p)$ answers $h=t$
directly. If $H_0\le p\le h<t$, pair it with the record at $p$; if $p=-1$,
it proves absence only for $t-\eta\le h<t$. A baseline record with
$p=H_0-1$ makes no claim below $t$. A touch therefore buys a short proof when its link or
sentinel interval reaches $h$, at the cost of an on-chain transaction.
With a prior-value touch record, the same predecessor interval needs only
the touch BUD: its prior value supplies the answer. An inherited-entry
marker also permits this single-BUD form, while a sentinel retains its
bounded absence interval (\cref{sec:prior-cert}).
SuperBUDs need no new transaction; \cref{sec:cert} gives the verification rules.

\begin{figure}[t]
\centering
\begin{minipage}[t]{0.34\linewidth}
\vspace{0pt}
\centering
\begin{tikzpicture}[
  x=1cm,y=1cm,
  root/.style={draw,rounded corners=1pt,inner sep=2pt,font=\scriptsize},
  leaf/.style={draw,inner sep=2pt,font=\tiny},
  note/.style={font=\tiny,align=center}
]
\node[root] (u) at (0,2.75) {$U_a$};
\node[leaf] (u1) at (-0.90,2.1) {$(x,v_x,a,p_x)$};
\node[leaf,fill=gray!25] (u2) at (0.90,2.1) {$(k,v,a,p)$};
\draw (u) -- (u1); \draw (u) -- (u2);
\node[note] at (0,1.72) {BUD: values and predecessor links};
\node[root] (sw) at (0,1.10) {$S_W$};
\node[leaf] (s1) at (-0.90,0.45) {$x\mapsto b$};
\node[leaf,fill=gray!25] (s2) at (0.90,0.45) {$k\mapsto a$};
\draw (sw) -- (s1); \draw (sw) -- (s2);
\node[note] at (0,0.07) {SuperBUD: last write in $W$};
\node[note] at (0,-0.23) {Leaves sorted by key ($x<k$).};
\end{tikzpicture}

{\scriptsize (a) Contents of the commitments.}
\end{minipage}\hfill
\begin{minipage}[t]{0.63\linewidth}
\vspace{0pt}
\centering
\begin{tikzpicture}[
  x=0.41cm, y=0.60cm,
  win/.style={draw, line width=0.25pt},
  mem/.style={fill=gray!35},
  exc/.style={fill=gray!15},
  lab/.style={font=\tiny, anchor=east},
  ann/.style={font=\tiny}
]
\fill[mem] (0,2.50)  rectangle (8,3.10);
\fill[exc] (8,1.67)  rectangle (12,2.27);
\fill[exc] (12,0.83) rectangle (14,1.43);
\fill[exc] (14,0.00) rectangle (15,0.60);
\foreach \i in {0,...,15} \draw[win] (\i,0.00) rectangle (\i+1,0.60);
\foreach \j in {0,...,7}  \draw[win] (2*\j,0.83) rectangle (2*\j+2,1.43);
\foreach \j in {0,...,3}  \draw[win] (4*\j,1.67) rectangle (4*\j+4,2.27);
\foreach \j in {0,1}      \draw[win] (8*\j,2.50) rectangle (8*\j+8,3.10);
\draw[win] (0,3.33) rectangle (16,3.93);
\node[lab] at (-0.3,0.30) {$\ell{=}0$};
\node[lab] at (-0.3,1.13) {$1$};
\node[lab] at (-0.3,1.97) {$2$};
\node[lab] at (-0.3,2.80) {$3$};
\node[lab] at (-0.3,3.63) {$4$};
\node[ann] at (4,2.80)  {mem., $k\mapsto a$};
\node[ann] at (10,1.97) {excl.};
\node[ann] at (13,1.13) {ex.};
\draw[dashed,line width=0.25pt] (3.5,-0.47)  -- (3.5,4.08);
\draw[dashed,line width=0.25pt] (14.5,-0.47) -- (14.5,4.08);
\draw[dotted] (8,-0.47) -- (8,2.50);
\fill (3.5,0.30) circle (1.1pt);
\draw[-{Stealth[length=1.5mm]}] (0,-0.47) -- (16.6,-0.47);
\node[ann,anchor=north] at (0,-0.60)    {$H_0$};
\node[ann,anchor=north] at (3.5,-0.60)  {$a$};
\node[ann,anchor=north] at (8,-0.60)    {$s$};
\node[ann,anchor=north] at (14.5,-0.60) {$h$};
\end{tikzpicture}
{\scriptsize (b) Aligned windows over block height.}
\end{minipage}
\caption{Proving a historical read with baseline records. (a) A BUD supplies the value and
predecessor link at $a$; a SuperBUD records each key's last write height
in its window. (b) The dark SuperBUD confirms that $a$ was the last write
before $s$; the lighter BUDs and SuperBUDs exclude $k$ through $h$. Together with
the BUD at $a$, they prove the value at $h$. Windows have lengths
$1,2,4,8,16$ ($e=2$, $L=4$), with BUDs at level $0$; larger maps keep
each key's latest height across their children.}
\label{fig:superbud}
\label{fig:bud}
\end{figure}

\subsection{SuperBUDs}
\label{sec:superbuds}

Without a successor record, BUDs alone need an exclusion proof for every
block since the last write, or the cost incurred by a touch transaction. A \emph{SuperBUD} instead allows proving exclusion or membership over a large window via summarizing many blocks at once. For an interval $W=[s,t]$, build a map containing every key modified in that
interval and the height of its last modification there. This is the
\emph{span map}
\[
M_W=\{k\mapsto\max(\Mod(k)\cap W):\Mod(k)\cap W\ne\emptyset\}.
\]
The SuperBUD $S_W$ is the Merkle commitment
$\MR_{\mathsf{SuperBUD},(s,t)}$ to the key-sorted entries of $M_W$.
Validators attest it after block $t$ is final. The map stores heights; values remain in the BUD records.

Excluding $k$ from $M_W$ proves that no block in $W$ modified it; opening
$k\mapsto a$ proves no change in $(a,t]$. Thus $k\mapsto10$ in a summary
of blocks $1$--$32$, paired with the BUD at $10$, proves the value through
$32$. Further summaries excluding $k$ extend that interval. The same rules
apply to deletion records.

\paragraph{Building a larger SuperBUD.}
For two sorted maps $A$ and $B$, scan their entries in key order. Emit a key
that occurs in only one map unchanged; when it occurs in both, emit the
larger height. Denote the resulting map by $A\vee B$. Since the last write
in a union is the latest write in any of its parts,
\[
M_{W_1\cup\cdots\cup W_r}=M_{W_1}\vee\cdots\vee M_{W_r}.
\]
The operation is associative, commutative, and idempotent, so overlapping
intervals cause no difficulty. Merging two maps is linear in their combined
length. Merging $e$ sorted child maps takes $O(s\log e)$ comparisons for
$s$ input entries, hence linear work for a fixed hierarchy base $e$.
Validators hash the merged entries to obtain the new root; child roots alone
do not suffice. A summary retains only the last write to each key, so it
cannot in general be used to recover a smaller arbitrary subinterval.

\subsubsection{Window schedules.}
\label{sec:schedule}

A schedule chooses which block intervals receive SuperBUDs. Our default is
the aligned hierarchy in \cref{fig:superbud}. For a fixed integer base
$e\ge2$ and depth $L$ with $e^L\le\eta$, define
\[
\window{\ell}{j}=[H_0+je^\ell,\ H_0+(j+1)e^\ell-1],
\quad 0\le\ell\le L,\quad j\ge0.
\]
Level $0$ consists of BUDs. A level-$(\ell+1)$ window contains exactly
$e$ level-$\ell$ windows, so its map is obtained by merging those children.
At each level $\ell\ge1$, one SuperBUD summarizes every $e^\ell$ blocks. Beyond the BUD,
the average number of new digests per block is therefore
\[
\sum_{\ell=1}^{L}e^{-\ell}
=\frac{1-e^{-L}}{e-1}.
\]
Each write contributes at most one entry per level; repeated writes to the
same key can reduce that count.

\paragraph{Proving an unchanged interval.}
Let $a$ be the last modification of $k$ at or before the query height $h$,
and set $d=h-a\ge1$. Choose
$\ell'=\min\{\lfloor\log_e d\rfloor,L\}$. If $\ell'>0$, take the
SuperBUD whose interval contains $a$. Its length is at most $d$, so its
endpoint $t$ is below $h$. Its map contains $k\mapsto a$, proving no
change in $(a,t]$. If $\ell'=0$, use the BUD at $a$ and set $t=a$.
Starting at $t+1$, take aligned BUDs or SuperBUDs covering $e^{\ell'}$ blocks
wherever they fit before $h+1$, then repeat at smaller levels, using BUDs
at level $0$. Each selected BUD or SuperBUD excludes $k$. This is the \emph{staircase} shown in \cref{fig:superbud}.

The remaining length is $R=h-t\le d$. The number of these BUDs and SuperBUDs is given by
the base-$e$ digits of $R$, with an unrestricted top digit when the
hierarchy is capped. Thus the number of exclusions is at most
\begin{equation}
\label{eq:staircase}
B(d)=
\begin{cases}
(e-1)\bigl(\lfloor\log_e d\rfloor+1\bigr), & d<e^{L+1},\\
\lfloor d/e^L\rfloor+(e-1)L, & d\ge e^{L+1}.
\end{cases}
\end{equation}
Including the BUD at $a$ and a distinct first SuperBUD when needed gives
at most $B(d)+2$ digests. Every selected BUD or SuperBUD summarizes
blocks at or below $h$: the proof needs no later write or block, only
retrieval of its records and attestations.

One SuperBUD can sometimes replace the whole staircase. If its interval
contains $a$ and $h$ and its entry for $k$ is still $a$, it proves the entire
unchanged interval. This is not guaranteed by alignment: an $a$ immediately
before a top-level boundary also lies before a boundary at every smaller
level. The staircase still works. \Cref{app:schedules} examines schedules
that trade more summaries for more frequent short proofs.

\subsection{Authenticated-map construction}
\label{sec:maps}

BUDs and SuperBUDs use the same Merkle encoding. Let $T$ identify the object
type, $I$ its height or interval, and $\mathbf z=(z_1,\ldots,z_m)$ its
strictly key-sorted entries. The tags $\mathtt{leaf}$, $\mathtt{pad}$,
$\mathtt{node}$, $\mathtt{root}$, and $\mathtt{empty}$ distinguish hash
inputs of different kinds. For $m>0$, pad to $M=2^{\lceil\log_2m\rceil}$
leaves and set
\[
x_i=
\begin{cases}
H(\mathtt{leaf}\parallel\enc(m,i,z_i)), &1\le i\le m,\\
H(\mathtt{pad}\parallel\enc(m,i)), &m<i\le M.
\end{cases}
\]
An internal node with children $x,y$ is
$H(\mathtt{node}\parallel x\parallel y)$. If $R$ is the resulting tree
root, set
\[
\begin{aligned}
\MR_{T,I}(\mathbf z)&=H(\mathtt{root}\parallel\enc(T,I,m,R)),\\
\MR_{T,I}(())&=H(\mathtt{empty}\parallel\enc(T,I,0)).
\end{aligned}
\]
An inclusion proof gives $m>0$, a position $1\le i\le m$, its entry, and
exactly $\lceil\log_2m\rceil$ siblings; the bits of $i-1$ determine the
left--right choices. Exclusion of $k$ requires one of four checks: the
distinguished empty root; position $1$ with a key greater than $k$;
position $m$ with a key less than $k$; or consecutive positions $i,i+1$
with keys $k'<k<k''$. All openings share the same type, index, root, and
leaf count. These checks prevent an archive from presenting interior entries
as boundaries or nonadjacent entries as neighbors.

\subsection{Attestation and read proofs}
\label{sec:authenticated-reads}

BUDs and SuperBUDs use the same attestation rule. Once their roots are
authenticated, a read reduces to opening a value and showing that no later
write changed it before the requested height.

\subsubsection{Attestation.}
\label{sec:attest}

A BUD $U_n$ is signed after block $n$ is final; a SuperBUD $S_W$ for
$W=[s,t]$ is signed after block $t$ is final. In either case, $\rho(D)$
denotes the committee epoch assigned to the last block summarized by the
digest $D$. An honest validator signs only the digest it derives from the
finalized chain, at most once for each BUD height or SuperBUD interval.

The signed message is
\[
m_D=\enc\!\left(
 \mathsf{BUDAttest},
 \mathsf{chainID},
 \mathsf{protocolVersion},
 \rho(D),
 \mathrm{type}(D),
 \mathrm{index}(D),
 D
\right).
\]
An attestation contains signatures from at least $\tau$ distinct members
of $\Pi_{\rho(D)}$, using the threshold specified in \cref{sec:model}.
The verifier checks their membership and signatures; an aggregate must bind
the exact signer set. A separate state-agreement digest may be included in
the same signed message without changing these conditions.

The verifier authenticates the committee keys through the chain's usual
checkpoint or light-client mechanism. We assume that mechanism accepts only
the canonical committee, and that the bound of $f$ compromised keys holds
for every historical epoch whose signatures are accepted. A deployment that
cannot preserve that bound needs key erasure, forward-secure signatures, or
a trusted later checkpoint that authenticates the historical digests.

\subsubsection{Read certificates.}
\label{sec:cert}

\begin{figure}[t]
\centering
\begin{tikzpicture}[
  x=0.42cm, y=0.82cm,
  ax/.style={-{Stealth[length=1.4mm]}},
  dig/.style={draw, rounded corners=1.5pt, inner sep=2pt, font=\tiny, fill=gray!25},
  tk/.style={font=\tiny, anchor=north},
  ttl/.style={font=\scriptsize, anchor=west},
  br/.style={decorate, decoration={brace, amplitude=2.5pt, mirror}}
]

\begin{scope}[shift={(0,0)}]
  \node[ttl] at (0,1.42) {(a) \textsc{Empty}};
  \draw[ax] (0,0) -- (12.4,0);
  \node[dig] (ua) at (6,0.72) {$U_a$};
  \draw[thin] (ua) -- (6,0.08);
  \fill (6,0) circle (1.4pt);
  \node[tk] at (6,-0.10) {$a=h$};
\end{scope}

\begin{scope}[shift={(15.2,0)}]
  \node[ttl] at (0,1.42) {(b) \textsc{Next}};
  \draw[ax] (0,0) -- (12.4,0);
  \node[dig] (ub1) at (2,0.72) {$U_a$};
  \node[dig] (ub2) at (10,0.72) {$U_b$};
  \draw[thin] (ub1) -- (2,0.08);
  \draw[thin] (ub2) -- (10,0.08);
  \draw[-{Stealth[length=1.4mm]}] (ub2.north west) to[out=140,in=40]
    node[midway, above, font=\tiny] {$\prev = a$} (ub1.north east);
  \fill (2,0) circle (1.4pt);
  \fill (6,0) circle (1.4pt);
  \fill (10,0) circle (1.4pt);
  \node[tk] at (2,-0.10) {$a$};
  \node[tk] at (6,-0.10) {$h$};
  \node[tk] at (10,-0.10) {$b$};
\end{scope}

\begin{scope}[shift={(0,-2.2)}]
  \node[ttl] at (0,1.42) {(c) \textsc{Cover}};
  \draw[fill=gray!35] (0.6,0.55) rectangle (3.4,1.05);
  \draw[fill=gray!12] (3.6,0.55) rectangle (7.4,1.05);
  \draw[fill=gray!12] (7.6,0.55) rectangle (9.4,1.05);
  \draw[fill=gray!12] (9.6,0.55) rectangle (10.4,1.05);
  \node[font=\tiny] at (2.0,0.80) {$k \mapsto a$};
  \node[font=\tiny] at (5.5,0.80) {excl.};
  \node[font=\tiny] at (8.5,0.80) {excl.};
  \node[font=\tiny] at (2.0,0.28) {$S_W$};
  \draw[ax] (0,0) -- (12.4,0);
  \fill (1,0) circle (1.4pt);
  \fill (10,0) circle (1.4pt);
  \node[tk] at (1,-0.10) {$a$};
  \node[tk] at (10,-0.10) {$h$};
  \node[dig] (uc) at (1,-0.72) {$U_a$};
  \draw[thin] (uc) -- (1,-0.12);
\end{scope}

\begin{scope}[shift={(15.2,-2.2)}]
  \node[ttl] at (0,1.42) {(d) sentinel};
  \draw[ax] (0,0) -- (12.4,0);
  \node[dig] (ud) at (10,0.72) {$U_b$};
  \draw[thin] (ud) -- (10,0.08);
  \node[font=\tiny, anchor=south] at (10,1.05) {$p = -1$};
  \fill (2,0) circle (1.4pt);
  \fill (6,0) circle (1.4pt);
  \fill (10,0) circle (1.4pt);
  \node[tk] at (2,-0.10) {$b-\eta$};
  \node[tk] at (6,-0.10) {$h$};
  \node[tk] at (10,-0.10) {$b$};
\end{scope}

\end{tikzpicture}
\caption{Four baseline certificate forms. (a) The anchor answers at $h=a$.
(b) A successor brackets $h$. (c) Span proofs clear $(a,h]$ while $U_a$
supplies the value. (d) A sentinel certifies bounded absence. All roots are
attested; the text specifies the acceptance checks. With prior-value
records, panel (b) needs only the successor BUD $U_b^+$.}
\label{fig:cert}
\end{figure}

An ordinary certificate has an \emph{anchor}: an inclusion path for
$(k,v,a,p)$ against an attested BUD $U_a$, with $H_0\le a\le h$; and an
\emph{extension} ruling out every modification in $(a,h]$. These anchors
cover $Q=\{(k,h):h\ge H_0,\ \Mod(k)\cap[H_0,h]\ne\emptyset\}$, the
queries for which completeness is guaranteed.

Three extensions are available (illustrated in \cref{fig:cert}):
\begin{enumerate}[leftmargin=3.6em,itemsep=1pt,align=left,labelwidth=3.2em,labelsep=0.4em,label=]
\item[(\textsc{Empty})] Nothing, when $h = a$. No extension needed.
\item[(\textsc{Next})] The next write to $k$, an inclusion path for a record
$(k,v',b,a)$ against an attested BUD $U_b$ with $b > h$. Its pointer refers to $a$, so no
block between the two modified $k$.
\item[(\textsc{Cover})] Paths against scheduled SuperBUDs and individual BUDs.
An exclusion path clears a digest's whole interval lying above $a$. An inclusion
path for $k\mapsto a$ in a SuperBUD clears the part of its interval above
$a$, because $a$ is the last modification there. The cleared parts must
together cover $(a,h]$.

\end{enumerate}
Exclusion claims admit a second kind of anchor, which needs no extension. A
record $(k,v,b,-1)$ is emitted only when the entry for $k$ is genuinely absent
immediately before block $b$. By \cref{lem:chain}(ii), this proves
$\val(k,h)=\absent$ for every
$h\in[b-\eta,b)\cap[H_0,b)$. We call this a \emph{sentinel anchor}. A record
carrying $H_0-1$ is an activation record, not a sentinel. In the baseline
format it makes no claim below its own height.

$\Verify$ authenticates each digest under the committee assigned to its
last included block and checks its signed type, index, and root. It then verifies
every map opening, including leaf counts, positions, path lengths, and
left--right choices, using the encoding in \cref{sec:maps}. All components
must name the same query key. An ordinary anchor must have
$H_0\le a\le h$ and forces $y=\out(v)$; its extension must establish that
every height in $(a,h]$ is unchanged. A sentinel must open $(k,v,b,-1)$ in
$U_b$, satisfy $b-\eta\le h<b$, and force $y=\absent$. The verifier rejects
$h<H_0$, unauthenticated committees, and any attempt to use $H_0-1$ as an
ordinary predecessor or a sentinel. The prior-value form below handles
the inherited-entry marker separately and never treats it as a BUD height.

\subsubsection{Single-BUD certificates with prior values.}
\label{sec:prior-cert}

In the prior-value variant, a certificate may instead consist solely of
an inclusion opening for $(k,v^-,v,b,p)$ in $U_b^+$ and its authentication
evidence. The verifier checks the enriched object type, protocol version,
height, query key, opening, and attestation as above. It rejects $h<H_0$.
For this certificate form, it accepts exactly the following cases,
requiring the supplied claim $y$ to equal the stated result:
\begin{enumerate}[label=(\roman*),leftmargin=2em,itemsep=1pt]
\item $h=b$: $y=\out(v)$, as in S1.
\item $H_0\le p\le h<b$: $y=\out(v^-)$.
This is the \emph{prior-value successor} strategy, S3-PV.
\item $p=H_0-1$ and $H_0\le h<b$: $y=\out(v^-)$.
This authenticates the inherited value directly from the successor.
\item $p=-1$, $v^-=\tomb$, and $\max\{H_0,b-\eta\}\le h<b$:
$y=\absent$, as in S2.
\end{enumerate}
Case (ii) needs no record or attestation at $p$: the successor commits both
the earlier value and the link proving no intervening write. Case (iii)
similarly uses the honest signer's validation of the inherited prior value.
Case (iv) does not regain history lost through tombstone collection;
an absent prior value alone cannot justify extending the sentinel interval.
No case answers $h>b$ without further evidence.

These certificates require the write at $b$ and its evidence to become
available, requiring a finalized height of at least $b$, or $b-h$ blocks
after the query height, before attestation and retrieval. Without such a record, the
existing anchor and cover rules still apply. In particular, the variant
does not give completeness for untouched inherited or never-written keys
without a suitable later write. Soundness and the original
completeness guarantee hold for both formats (\cref{app:correctness}).

\subsubsection{Choosing a proof strategy.}
\label{sec:frontier}

We evaluate five baseline strategies that can be followed to generate proofs: anchor only (S1), sentinel (S2), successor
link (S3), one SuperBUD (S4), and staircase (S5). S1 uses \textsc{Empty},
S3 uses \textsc{Next}, and S4 and S5 are two ways to satisfy \textsc{Cover};
S2 needs no extension. Here $a$ is the last write by $h$, and $d=h-a$.
\begin{table}[t]
\caption{Proof strategies and optional touch transactions. Counts are
distinct attested digests; waits exclude ordinary attestation and retrieval.
The staircase bound includes the anchor and, when distinct, its first
SuperBUD. S3-PV and Touch--prior value are optional, unbenchmarked variants.}
\label{tab:frontier}
\centering
\footnotesize
\setlength{\tabcolsep}{3pt}
\begin{tabular}{@{}>{\raggedright\arraybackslash}p{2.9cm}>{\raggedright\arraybackslash}p{1.85cm}>{\raggedright\arraybackslash}p{1.8cm}>{\raggedright\arraybackslash}p{4.9cm}@{}}
\toprule
Strategy & Digests & Extra wait & Applies when \\
\midrule
Anchor only (S1) & $1$ & none & $h=a$ \\
Sentinel (S2) & $1$ & none & $p=-1$ and $b-\eta\le h<b$ \\
Successor link (S3) & $2$ & next write & successor points to $a$ \\
Prior-value successor (S3-PV) & $1$ & next write & successor commits $v^-$ and points to $a$ \\
One SuperBUD (S4) & $2$ & interval end & one SuperBUD proves no change \\
Staircase (S5) & \mbox{$\le B(d)+2$} & none & $(k,h)\in Q$, $d\ge1$ \\
\midrule
Touch--\textsc{Next} & $2$ & touch finality & emitted $H_0\le p=a\le h<t$ \\
Touch--prior value & $1$ & touch finality & emitted $H_0-1\le p\le h<t$, with $v^-$ committed \\
Touch--sentinel & $1$ & touch finality & $p=-1$ and $t-\eta\le h<t$ \\
Touch at target & $1$ & touch finality & $h=t$ \\
\bottomrule
\end{tabular}
\end{table}

All five strategies use existing evidence without a new transaction.
S1 applies to every query in $Q$ with $a=h$, and S5 to all others. S3
needs a successor pointing directly to $a$. S3-PV uses the same successor
when prior values are committed, reducing its digest count from two to one
without changing its wait. S4 needs one SuperBUD that
either opens $k\mapsto a$ in an interval containing $a$, or excludes $k$
over $(a,h]$; any heights beyond $h$ in its interval must also be unchanged.
A touch can create S1, S2, or S3 under \cref{sec:touch}'s conditions.
Waiting for a successor has no deadline; alternative schedules bound the
wait for a short proof over a configured range of gaps (\cref{app:schedules}).

\paragraph{Guarantees.}
\label{sec:guarantees}
Under the model assumptions, every accepting certificate returns the
canonical value, and every query in $Q$ has an accepting certificate once
its records, attestations, and committee evidence are available. Neither
guarantee requires a bound on message delay. The soundness and completeness
proofs appear in \cref{app:correctness}. Outside $Q$, an inherited key's first
baseline record makes no backward claim. Its first prior-value record can
authenticate the inherited value for $H_0\le h<b$
(\cref{sec:prior-cert}). A later sentinel can certify absence only
when $h<b\le h+\eta$.

\subsection{Costs}
\label{sec:costs}

\paragraph{Validator work.}
A BUD sorts $w_n$ records and performs $O(1+w_n)$ hashes. Once sorted, the
hashing reads the write buffer and takes $O(\log(2+w_n))$ parallel tree levels;
it does not traverse a tree over the live state. The hierarchy adds at most
one entry per write per level. For fixed $e$, merging and hashing therefore
add $O(L)$ amortized work per write. This is an amortized bound: several
large SuperBUDs can be built at the same height, and a top-level SuperBUD may
summarize $O(we^L)$ distinct keys when each block writes at most $w$ keys.
The evaluation reports tails with this maintenance enabled.

Validators keep flat state plus an $8$-byte predecessor field per retained
entry. They also retain child maps or recent logs for unfinished summaries;
$O(we^L)$ map entries suffice for the aligned hierarchy when blocks contain
at most $w$ writes. Tombstones can be placed in expiry buckets and removed in $O(1)$
amortized work per deletion; collection does not relocate live entries to
preserve an authenticated structure. Ordinary state reads and writes still
depend on the database and cache. Archives retain BUD records, SuperBUD
maps, Merkle material, and attestations. Across $L$ levels, the maps contain
at most $L$ extra entries per logged write; this is an entry-count bound,
not a bound on serialized bytes.

Prior-value records add one encoded value per logged write to the write
buffer and archive, but no entry to a SuperBUD map. With fixed-width
$q$-byte prior values, the extra record payload is $q w_n$ bytes per block;
for a $32$-byte storage-slot value, this is $32$ bytes per record, apart from
any additional encoding tags. The BUD root size, leaf count, and Merkle
path depth are unchanged. Hashing and copying process the extra value bytes;
their cost must be included when values have variable length. Persistent
predecessor metadata remains unchanged.

\paragraph{Proof work and bytes.}
Opening the anchor at $a$ takes $O(\log(2+w_a))$ hashes. A successor adds
one such path at $b$; a cover adds paths costing
$O(\sum_i\log(2+s_i))$, where $s_i$ is the number of entries in its $i$th
map. An interior exclusion opens two adjacent paths, whose shared hashes
can be deduplicated. The staircase uses at most $B(d)+2$ digests; counting
digests alone does not count the hashes within each path.

Our reported proof bytes contain the records, map entries, identifiers,
and Merkle paths. A transmitted proof also needs an attestation for each
distinct digest: a signature or aggregate plus its signer set. It may also
need committee or checkpoint evidence not already cached by the reader.
Those chain-specific bytes are excluded from the measurements, as is
signature verification from the reported hash-path timings.

\paragraph{Savings from a prior-value successor.}
For the same query and successor, let $P_a,P_b$ be the baseline BUD
inclusion payloads, including records, paths, and digest identifiers. Let
$A_a,A_b$ be the separately supplied per-digest attestation bytes, and let
$q_b$ be the extra encoded prior-value bytes at $b$, including any format
overhead. With otherwise equal encoding widths, the opening and separate
attestation bytes, excluding common committee evidence, are
\[
\begin{aligned}
C_{\mathrm{S3}}&=P_a+P_b+A_a+A_b,\\
C_{\mathrm{S3\text{-}PV}}&=P_b+q_b+A_b,\\
C_{\mathrm{S3}}-C_{\mathrm{S3\text{-}PV}}&=P_a+A_a-q_b.
\end{aligned}
\]
The saving is positive exactly when the omitted opening and attestation
exceed the added prior value. For fixed-width values, $q_b$ does not grow
with block size, whereas the omitted path grows logarithmically with $w_a$.
There is consequently no universal break-even block size or percentage
saving. If the two separately attested openings have similar size and
$q_b$ is small by comparison, this portion of the certificate approaches
half the size; shared overhead can reduce the fractional saving overall.
Other certificate forms can grow: each opened BUD record carries its prior
value, including the neighboring records in exclusion openings.

With individual signatures, omitting an attestation removes a quorum
signature bundle. Caching evidence or aggregating signatures across digests
can reduce the authentication bytes saved. Aggregation remains optional,
including for a post-quantum deployment
using individual signatures. Its hashes and signatures must satisfy that
deployment's security assumptions. The byte comparison is analytical;
the prior-value variant has not been benchmarked.

\subsection{Deployment and availability}
\label{sec:limits}

\paragraph{Activating existing keys.}
When an existing chain switches to BUDs at $H_0$, validators hold
$\sigma_0$, but BUDs record only new writes. An inherited key left unchanged
has no BUD anchor: $H_0-1$ is a boundary marker, not proof of its value.
Its first write or touch at $a\ge H_0$ creates an anchor for reads from
$a$ onward. In the prior-value format, that record also authenticates the
inherited value for $H_0\le h<a$. This adds retrospective coverage once
the record is available; it does not answer reads of keys that remain untouched.

Deployment can wait for ordinary writes, with no deadline for untouched
keys, or deterministically scan inherited keys and log value-preserving
writes. The scan skips already-activated keys and must track activation
status or postpone tombstone collection until each key has been considered.
Scanning $N$ keys in chunks of $w$ takes about $N/w$ blocks: $5{,}000$ for
$N=10^8$ and $w=2\cdot10^4$. Once activation is attested, individual BUDs
and SuperBUDs ending by the query height suffice; no hierarchy warm-up is
needed. Earlier reads can use archived state-root proofs during migration.

\paragraph{Joining validators and serving proofs.}
New validators need canonical state, predecessor metadata, and unfinished
SuperBUD maps or logs to reconstruct them. State-sync or checkpoint
mechanisms supply and authenticate these objects; BUDs do not provide
an initial-state commitment. Archives separately retain BUD records, SuperBUD maps,
Merkle paths, attestations, and committee evidence. Missing data can prevent
a historical read, but cannot make an incorrect proof verify. Archive
replication and incentives remain deployment choices.


\section{Evaluation}
\label{sec:eval}
We measure base-BUD replay as state grows, historical-read proofs, and
SuperBUD maintenance separately. Paired longer runs measure QMDB with
garbage collection dormant and active. All measurements here and in
\cref{app:eval:q3} use the baseline record format $(k,v,n,p)$. Prior-value
records and S3-PV have not been implemented or benchmarked; reported S3
sizes, timings, and strategy shares refer to the two-BUD baseline.

\subsection{Setup}

All systems replay the same blocks. BUD applies writes to flat state
and builds each BUD. We compare it with a simplified in-memory Merkle Patricia
trie, a disk-backed version, NOMT~\cite{nomt}, and QMDB~\cite{qmdb}.
The prototype trie omits RLP and Ethereum's account/storage split. Synthetic
traces vary the state size $N$ and writes per block $w$, and include deletions.
We also replay an Ethereum account-access trace covering $24{,}576$ blocks,
with $445$ marked accounts per block on average. Accounts identified from
public RPC data are treated as writes and mapped to dense identifiers. This
preserves their recurrence patterns but does not reproduce Ethereum
execution: marking can include unchanged accounts and miss internal changes,
and storage writes, deletions, and actual values are unavailable
(\cref{app:eval:trace}).

Commitment points comprise $128$ measured blocks on an Apple M4 with
$16$\,GB RAM. Only the disk-backed trie is cache-capped; NOMT and QMDB use
their defaults, so absolute comparisons are configuration-specific.
Certificate experiments run for $6\eta$ blocks and build every applicable
strategy for each sampled anchor and staleness. We sample both write events,
which favor recurring keys, and observed keys uniformly, which give rarely
written keys equal weight. Verification times cover warm hash paths. The
prototype implements replay, BUDs, SuperBUDs, and hash proofs.
\Cref{app:eval:q3:setup} gives the complete measurement boundaries.

\subsection{Commitment cost}
\label{sec:eval:commitcost}

\begin{figure}[t]
\centering
\sbox\frlbox{\includegraphics[width=0.98\evcolL]{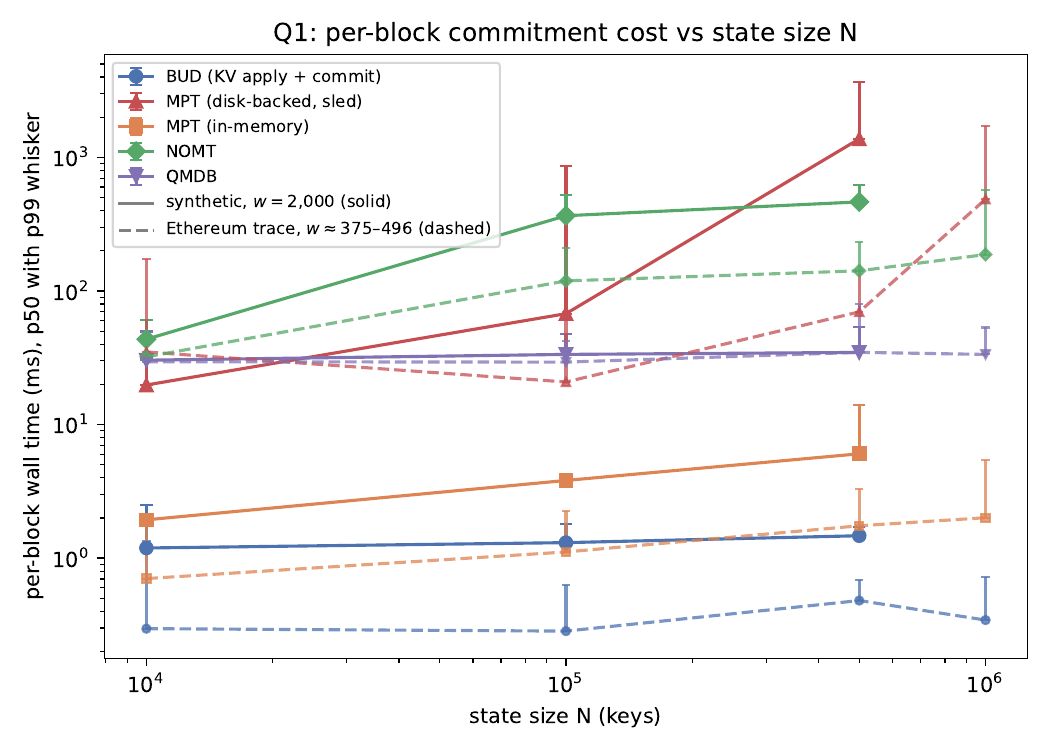}}
\sbox\frrbox{\scriptsize
\setlength{\tabcolsep}{1.0pt}
\renewcommand{\arraystretch}{1.35}
\begin{tabular}{@{}lccc@{\hspace{3pt}}c@{}}
\toprule
& \multicolumn{3}{c}{Synthetic} & Ethereum \\
\cmidrule(lr){2-4}\cmidrule(l){5-5}
Backend & $w{=}10^2$ & $w{=}2{\cdot}10^3$ & $w{=}10^5$ & $w{\approx}375$ \\
\midrule
BUD & 0.77   & 0.67   & 0.89 & 0.75   \\
MPT & 7.26   & 2.23   & 1.23 & 2.74   \\
NOMT         & 356.44 & 123.37 & 5.46 & 295.49 \\
QMDB         & 322.24 & 17.17  & 0.80 & 74.82  \\
\bottomrule
\end{tabular}}
\frmeasure
\begin{minipage}[t]{0.53\linewidth}
\vspace{0pt}
\centering
\frslot{\usebox\frlbox}
{\small (a) State-size sweep.}
\end{minipage}\hfill
\begin{minipage}[t]{0.45\linewidth}
\vspace{0pt}
\centering
\frslot{\usebox\frrbox}
{\small (b) Mean replay time per key.}
\end{minipage}
\caption{Commitment-cost results. (a) Per-block p50 versus state size (p99
whiskers): solid is synthetic ($w=2{,}000$), dashed is $128$-block Ethereum
trace windows ($w\approx375$--$496$). Only the disk MPT is
cache-capped. (b) Mean replay time per key at $N=10^5$, in
$\mu\mathrm{s}$; BUD includes state application and BUD construction,
but excludes SuperBUD maintenance. Trace scope and measurement boundaries are
in \cref{app:eval:q3:setup}.}
\label{fig:eval:nsweep}
\label{fig:eval:wsweep}
\end{figure}

\paragraph{Growing state, fixed writes.}
Growing the synthetic state from $10^4$ to $5{\cdot}10^5$ keys at
$w=2{,}000$ raises BUD time from $1.19$ to $1.47$\,ms
($1.24\times$), versus $3.1\times$ for the in-memory trie and $69.5\times$
for the disk trie. NOMT rises from $43.6$ to $464.8$\,ms; QMDB is nearly
flat (\cref{fig:eval:nsweep}(a)), but its garbage collector is dormant
in this sweep; the active-collection costs appear in \cref{tab:eval:q7}
and the paragraph below. BUD authenticates the same write volume
throughout, although flat-state access can still become more expensive.
On the Ethereum trace, BUD grows by $1.16\times$ from $10^4$ to
$10^6$ keys, versus $13.9\times$ for the disk trie. The disk results reflect
its $16$\,MiB cache cap, not a universal penalty for disk-backed storage.

\paragraph{Changing the block size.}
At $N=10^5$, BUD takes $0.67$--$0.89\,\mu\mathrm{s}$ per
synthetic write across $w=10^2$ to $10^5$
(\cref{fig:eval:wsweep}(b)); the in-memory trie improves from $7.26$ to
$1.23\,\mu\mathrm{s}$. The gap is largest for sparse blocks:
$9.4\times$ at $100$ writes versus $1.4\times$ at $100{,}000$.
Large batches amortize trie maintenance, while BUD remains inexpensive
for small updates. On the account-access trace, BUD takes
$0.75\,\mu\mathrm{s}$ per marked account, against the trie's $2.74$.

\paragraph{QMDB: the cost of maintaining the commitment.}
QMDB is competitive in the short commitment sweep, but garbage collection
remains dormant. On paired traces, lowering its activation threshold raises
process CPU time by $1.3$--$2.7\times$ and device bytes written by
$1.6$--$2.4\times$ per logical write relative to dormant QMDB
(\cref{tab:eval:q7}). The relocation code and durability policy are
unchanged. This extra work persists even where median block latency improves.

BUDs and SuperBUDs authenticate immutable write history: garbage collection
in the local state database does not require relocating live authenticated
entries or rebuilding these commitments. They pay for block digests and
SuperBUD merges (\cref{sec:eval:q5}) without coupling those commitments to
physical entry placement. BUDs avoid QMDB's authenticated live-entry
relocation work, whose cost the paired runs measure under unchanged
durability. BUD replay keeps state in memory; QMDB persists it. Their replay
times therefore do not measure throughput under equal storage guarantees. \Cref{sec:eval:q7} reports all configurations, work and
latency results, the removable compactor stalls, and the retained-storage
experiments.

\subsection{Certificate cost and strategy choice}
\label{sec:eval:q3}

\begin{figure}[t]
\centering
\begin{minipage}[t]{0.49\linewidth}
  \centering
  \includegraphics[width=0.94\linewidth]{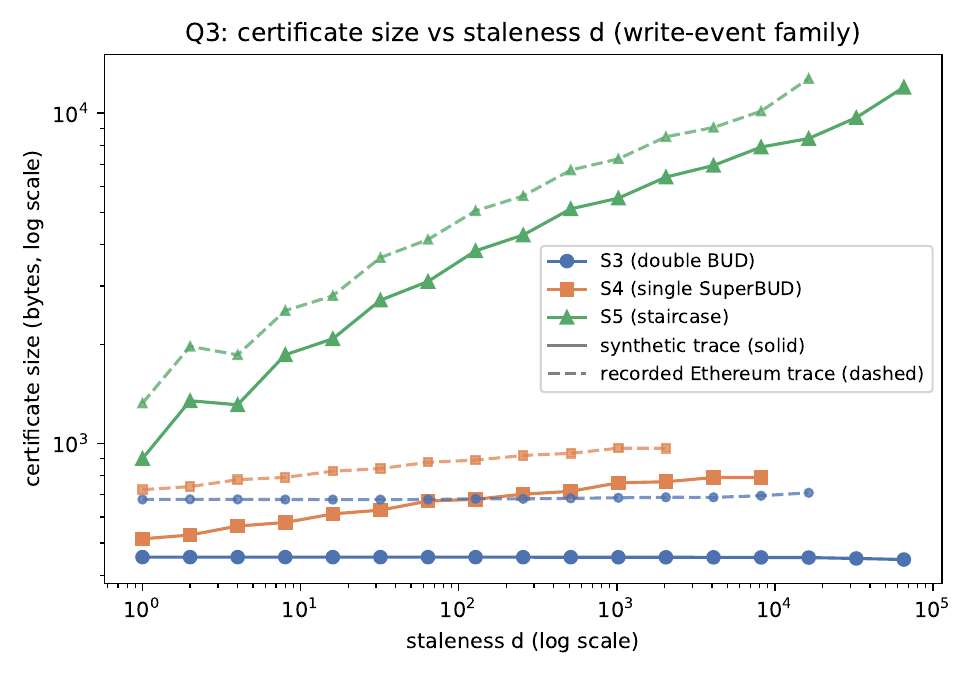}
  {\small (a) Certificate size.}
\end{minipage}\hfill
\begin{minipage}[t]{0.49\linewidth}
  \centering
  \includegraphics[width=0.94\linewidth]{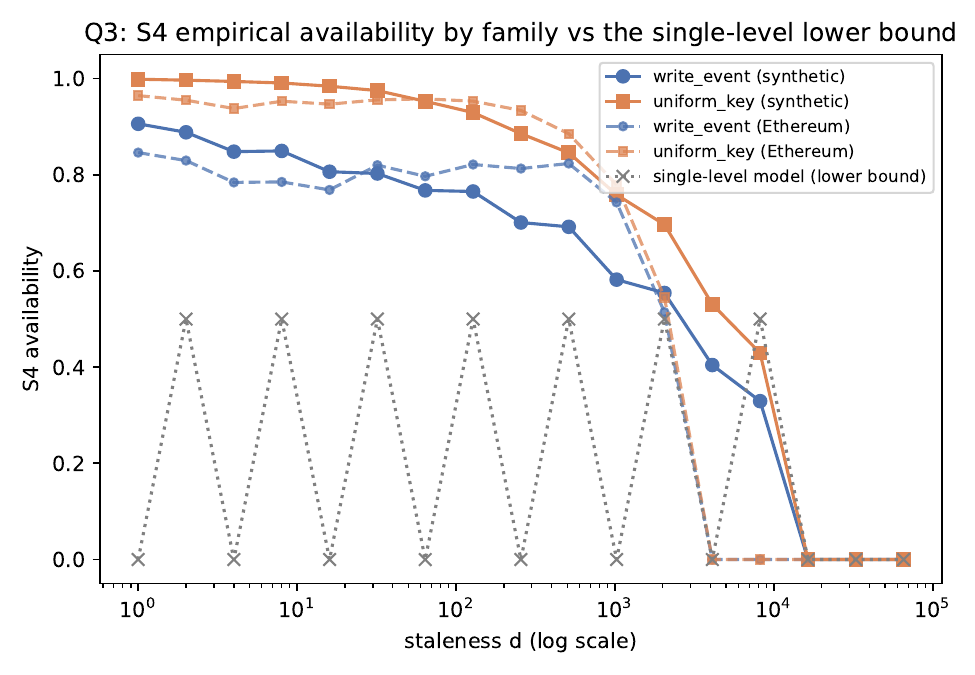}
  {\small (b) S4 availability.}
\end{minipage}
\caption{Certificate payload size and S4 availability versus staleness.
Solid: synthetic ($\eta=16{,}384$, $w=50$); dashed: Ethereum account-access
trace ($\eta=4{,}096$, full-trace mean $w\approx445$). Bytes use the benchmark
encoding; totals exclude attestation and committee evidence
(\cref{app:eval:q3:costs}).}
\label{fig:eval:q3cost}
\end{figure}

\paragraph{Proof size and verification.}
On the synthetic trace, the successor-link proof (S3) remains near $453$
bytes, the one-SuperBUD proof (S4) stays below $800$ bytes, and the
staircase proof (S5) grows logarithmically to about $12$\,kB. Every
measured S5 exclusion count satisfies
\cref{eq:staircase}. Warm hash-path verification reaches at most
$146\,\mu\mathrm{s}$ at p99 (\cref{fig:eval:q3cost}(a)). Attestation checks
are additional: a separate BLS12-381 benchmark, using preaggregated keys for
$67$ of $100$ authenticated signers, takes about $0.65$\,ms to batch-verify the two
aggregate signatures for S3 or S4 (\cref{app:eval:q3:setup}). Under the same
benchmark encoding, our prototype trie produces a $2{,}087$-byte inclusion
proof for the easier current-value query.

\paragraph{Proof availability by workload.}

At $d=1$, S3 is cheapest for $84.2\%$ of uniform-key synthetic queries
but $28.4\%$ on the Ethereum trace, where half the observed accounts appear
only once. S4 covers much of this gap but becomes less available as
staleness grows (\cref{fig:eval:q3cost}(b)); by $d=64$, both workloads need
about $5.5$ BUDs or SuperBUDs in the S5 cover on average. Rarely written keys therefore need
SuperBUDs even when active keys have short successor proofs. Sampling write
events instead raises the Ethereum S3 share at $d=1$ to $79.0\%$, showing
why event sampling alone underrepresents these queries. Complete costs,
shares, and skew sensitivity appear in
\cref{app:eval:q3:costs,app:eval:q3:strategies,app:eval:skew}.

\subsection{Schedules and hierarchy parameters}
\label{sec:eval:q5}

The base-BUD measurements exclude the cost of maintaining SuperBUDs.
\Cref{tab:eval:params} includes that work and shows the cost of shorter proofs.
With $N=10^5$, $w=50$, and $\eta=4{,}096$, increasing the branching factor
from $e=2$ to $e=8$ reduces the number of levels from $12$ to $4$. Extra
digests fall from nearly one per block to one per seven blocks, and p99 replay
time falls from $2.49$ to $1.24$\,ms. The cost shifts to readers: mean
S5 cover size rises from $4.5$ to $8.5$ BUDs or SuperBUDs, and S4 availability drops by
about ten percentage points. Measured digest rates match
$(1-e^{-L})/(e-1)$ at the reported precision. These results include hierarchy
maintenance, but still exclude signing and attestation.

\begin{table}[t]
\caption{Hierarchy trade-offs at $\eta=4{,}096$ and $d\in[64,128)$, with
$w=50$. Extra digests exclude the mandatory BUD; block p99 covers state
updates, BUD construction, and hierarchy maintenance; S5 cover size counts
the BUD or SuperBUD interval proofs, averaged over the stated staleness range.}
\label{tab:eval:params}
\centering
\small
\setlength{\tabcolsep}{4pt}
\begin{tabular}{@{}ccccc@{}}
\toprule
$(e,L)$ & \shortstack{Extra\\digests/blk}
& \shortstack{Block p99\\($\mu\mathrm{s}$)}
& \shortstack{S5\\cover size}
& \shortstack{S4\\avail.} \\
\midrule
$(2,12)$ & 0.999756 & 2{,}485.8 & 4.51 & 79.9\% \\
$(4,6)$  & 0.333252 & 1{,}517.9 & 6.01 & 76.5\% \\
$(8,4)$  & 0.142822 & 1{,}235.0 & 8.52 & 69.8\% \\
\bottomrule
\end{tabular}
\end{table}

\begin{figure}[t]
\centering
\begin{tikzpicture}[x=0.065cm,y=0.32cm,
  bar/.style={draw,fill=gray!15,line width=0.35pt},
  lab/.style={font=\scriptsize,anchor=east},
  tick/.style={font=\tiny,anchor=north}]
\node[lab] at (-3,1.8) {Aligned};
\draw[bar] (0,1.4) rectangle (64,2.2);
\draw[bar] (64,1.4) rectangle (128,2.2);
\node[lab] at (-3,-1.0) {Full striding};
\foreach \j in {0,...,4} {
  \draw[bar] (16*\j,-0.55*\j) rectangle (16*\j+64,0.42-0.55*\j);
}
\draw[-{Stealth[length=1.3mm]}] (0,-2.8) -- (131,-2.8);
\foreach \b in {0,16,32,48,64,80,96,112,128} {
  \draw (\b,-2.7) -- (\b,-2.9);
  \node[tick] at (\b,-3.0) {$\b$};
}
\node[font=\scriptsize,anchor=north] at (64,-3.65) {Block offset from $H_0$};
\end{tikzpicture}
\caption{Aligned and fully strided $64$-block SuperBUDs ($e=4$), starting
every $64$ or $16$ blocks. Bars show one hierarchy level; ticks mark interval boundaries.}
\label{fig:striding}
\end{figure}

SuperBUDs of the same size can also cover overlapping intervals.
\emph{Full striding} starts a SuperBUD spanning $e^\ell$ blocks every
$e^{\ell-1}$ blocks, instead of every $e^\ell$ blocks. For $e=4$, for
example, $64$-block SuperBUDs start every $16$ blocks
(\cref{fig:striding}). More starting positions let a single SuperBUD cover
more query intervals, at the cost of maintaining more summaries. At $e=4$
full striding produces $4.0\times$ as many digests as the aligned schedule and raises
synthetic S4 availability from $83.7\%$ to $96.0\%$. Archived bytes per block
rise from $4{,}252$ to $11{,}810$; the Ethereum trace shows the same direction
of change. Using only disjoint $64$-block SuperBUDs cuts extra digests by
about $21\times$ and more than halves archived bytes, but lowers S4
availability to $46.1\%$ and raises the S5 p99 cover from $12$ to $71$
BUDs or SuperBUDs. \Cref{app:eval:schedules} gives both workloads and the sampling
conditions. The intermediate, two-phase schedule has a bounded-wait
guarantee of at most two digests for gaps within its configured reach
(\cref{app:schedules}); it was not separately benchmarked. In the prior-value
variant, the successor case uses one digest, while the S4 case still uses two.


\clearpage
\bibliographystyle{splncs04}
\bibliography{refs}

\clearpage
\appendix
The appendices give additional related work, alternative SuperBUD
schedules, and experimental details. All proofs are collected in
\cref{app:correctness}. The construction and verification rules are in
\cref{sec:construction}.

\begin{table}[!ht]
\caption{Notation.}
\label{tab:terms}
\centering
\footnotesize
\renewcommand{\arraystretch}{1.10}
\begin{tabular}{@{}ll@{}}
\toprule
Symbol & Meaning \\
\midrule
BUD, $U_n$ & Commitment to the writes of block $n$. \\
$U_n^+$, $v^-$ & Prior-value BUD and value before a block's first write. \\
SuperBUD, $S_W$ & Commitment to the last-write map for interval $W$. \\
$\Delta_n$, $w_n$ & Block $n$'s write log and its number of distinct keys. \\
$N$, $H_0$ & Live-state size and first logging height. \\
$\Mod(k)$, $\val(k,h)$ & Logged modification heights and value at height $h$. \\
$h$, $a$, $d=h-a$ & Query height, last write by $h$, and intervening gap. \\
$p\ge H_0$ & Height of the preceding logged modification. \\
$p=H_0-1$ & First logged write to an inherited key. \\
$p=-1$ & No local entry before the write. \\
$\eta$ & Tombstone retention in blocks. \\
$e$, $L$, $\ell$ & Hierarchy base, maximum level, and current level. \\
$W$, $M_W$ & Block interval and its key-to-last-write map. \\
$\varphi$ & Number of staggered schedules at each level. \\
$f$, $\tau$ & Byzantine bound and attestation threshold. \\
\bottomrule
\end{tabular}
\end{table}

\section{Additional related work}
\label{app:related}

Solana mixes a lattice hash of all accounts into each bank hash, while Sui
commits an elliptic-curve multiset hash of live objects
\cite{simd0215,simd0223,suicheckpoints}. Both follow incremental multiset
hashing~\cite{bellare1997,lewi2019}. Ethereum's deferred-root and
post-execution access-list proposals reduce synchronous work, but the global
root remains; proving one access-list entry also requires the list
\cite{eip7862,eip7928}. Aardvark is the closest stateless design: validators
retain a commitment, untrusted servers prove, and a version window admits stale
proofs~\cite{aardvark}. Unlike BUDs, these commitments span the dictionary.

Key-transparency systems authenticate epoch values and append-only directory
evolution~\cite{coniks,seemless,parakeet,merkle2}; tamper-evident logs and
persistent dictionaries support membership at historical snapshots
\cite{crosby2009,pads2001,balloon}. Plasma Cash addresses interval exclusion
with RSA accumulators~\cite{plasmacash,rsaexclusion}. Verifiable ledger
databases instead assume a trusted operator~\cite{ledgerdb,glassdb}, while
Ethereum state expiry parallels our retention window and touches
\cite{stateexpiry}. SuperBUDs resemble history trees and Merkle Mountain Ranges
\cite{crosby2009,mmr}; their hierarchy authenticates bucketing analogous to
exponential histograms~\cite{dgim}, and staggered windows resemble sparse
tables for idempotent aggregates~\cite{bfc2000}.



\section{Alternative window schedules}
\label{app:schedules}

The aligned SuperBUD hierarchy keeps maintenance costs small, but its boundaries can
prevent a single SuperBUD from covering a query. Additional SuperBUDs make
short proofs available sooner, at the cost of more construction and archive
storage. We compare three alternatives in \cref{fig:schedules}. All retain
the per-block BUDs, and all use the same certificate checks as the main
construction. Waiting below means waiting for the last block in a
SuperBUD's interval to be finalized; obtaining its attestation and proof data is a
separate availability requirement.

Fix $\eta=e^L$, with $e\ge2$ and $L\ge1$. Let $a$ be the last modification
of the queried key at or below $h$, and put $d=h-a\ge1$. The rates below
count SuperBUDs per block after startup; every block also has a BUD. Storage
bounds count archived BUD records and SuperBUD entries, rather than bytes.

\begin{figure}[!htbp]
\centering
\begin{tikzpicture}[
  x=0.5cm, y=1cm,
  win/.style={draw, thin},
  bud/.style={draw, thin, fill=gray!25},
  ph/.style={draw, thin, fill=gray!8},
  lab/.style={font=\tiny, anchor=east},
  ttl/.style={font=\small, anchor=west},
  ann/.style={font=\tiny}
]

\begin{scope}[shift={(0cm,0cm)}]
  \node[ttl] at (-0.8,2.30) {(a) one size, $m=4$};
  \foreach \i in {0,...,7} \draw[bud] (\i,0) rectangle (\i+1,0.40);
  \foreach \j in {0,1}     \draw[win] (4*\j,1.10) rectangle (4*\j+4,1.50);
  \node[lab] at (-0.3,0.20) {$1$};
  \node[lab] at (-0.3,1.30) {$4$};
\end{scope}

\begin{scope}[shift={(6.4cm,0cm)}]
  \node[ttl] at (-0.8,2.30) {(b) aligned, $e=2$ ($\varphi=1$)};
  \foreach \i in {0,...,7} \draw[bud] (\i,0) rectangle (\i+1,0.40);
  \foreach \j in {0,...,3} \draw[win] (2*\j,0.55) rectangle (2*\j+2,0.95);
  \foreach \j in {0,1}     \draw[win] (4*\j,1.10) rectangle (4*\j+4,1.50);
  \draw[win] (0,1.65) rectangle (8,2.05);
  \node[lab] at (-0.3,0.20) {$1$};
  \node[lab] at (-0.3,0.75) {$2$};
  \node[lab] at (-0.3,1.30) {$4$};
  \node[lab] at (-0.3,1.85) {$8$};
\end{scope}

\begin{scope}[shift={(0cm,-3.2cm)}]
  \node[ttl] at (-0.8,2.30) {(c) staggered, $e=2$, $\varphi=2$};
  \foreach \i in {0,...,7} \draw[bud] (\i,0) rectangle (\i+1,0.40);
  \foreach \j in {0,...,3} \draw[win] (2*\j,0.55) rectangle (2*\j+2,0.73);
  \foreach \j in {0,1,2}   \draw[ph]  (2*\j+1,0.77) rectangle (2*\j+3,0.95);
  \foreach \j in {0,1}     \draw[win] (4*\j,1.10) rectangle (4*\j+4,1.28);
  \draw[ph]  (2,1.32) rectangle (6,1.50);
  \draw[win] (0,1.65) rectangle (8,1.83);
  \node[lab] at (-0.3,0.20) {$1$};
  \node[lab] at (-0.3,0.75) {$2$};
  \node[lab] at (-0.3,1.30) {$4$};
  \node[lab] at (-0.3,1.74) {$8$};
\end{scope}

\begin{scope}[shift={(6.4cm,-3.2cm)}]
  \node[ttl] at (-0.8,2.30) {(d) per-block sliding};
  \foreach \i in {0,...,7} \draw[bud] (\i,0) rectangle (\i+1,0.40);
  \foreach \i in {0,1,2} \draw[win] (\i,0.55+0.13*\i) rectangle (\i+2,0.66+0.13*\i);
  \node[ann] at (4.8,0.85) {$\cdots$};
  \foreach \i in {0,1,2} \draw[win] (\i,1.10+0.13*\i) rectangle (\i+4,1.21+0.13*\i);
  \node[ann] at (6.8,1.40) {$\cdots$};
  \draw[win] (0,1.65) rectangle (8,1.76);
  \node[ann] at (8.6,1.83) {$\cdots$};
  \node[lab] at (-0.3,0.20) {$1$};
  \node[lab] at (-0.3,0.85) {$2$};
  \node[lab] at (-0.3,1.40) {$4$};
  \node[lab] at (-0.3,1.70) {$8$};
\end{scope}

\end{tikzpicture}
\caption{Four schedules for the same write log. Each rectangle covers the blocks
summarized by one BUD or SuperBUD; the bottom row contains the BUDs. (a) One fixed
SuperBUD size. (b) Several sizes with aligned boundaries. (c) Additional
staggered windows reduce gaps between possible starts; at $e=2$, two
positions per size give full striding. (d) Sliding windows provide every
possible start, at the cost of rebuilding each size after every block.}
\label{fig:schedules}
\end{figure}

\subsection{One SuperBUD size}
Choose a size $m$ and summarize the disjoint intervals
$[H_0+jm,H_0+(j+1)m-1]$, for $j\ge0$. This adds one SuperBUD every $m$ blocks.
A proof uses size-$m$ SuperBUDs wherever possible and BUDs for the
uncovered blocks at the ends. It is cheap to maintain, but a long query
interval needs many SuperBUDs.

\begin{lemma}[Fixed-size SuperBUDs]
\label{lem:single-size}
\label{lem:altsched}
The rate is $1/m$ SuperBUDs per block. Each write contributes
at most one SuperBUD entry. A cover of $(a,h]$ uses at most
$\lfloor d/m\rfloor+2(m-1)$ BUDs or SuperBUDs, all ending by $h$; the anchor BUD is
additional. For a chosen distance $d$, this bound is minimized at
$m=\Theta(\sqrt d)$.
\end{lemma}

\subsection{Staggered SuperBUDs}
At level $\ell$, write $E=e^\ell$. The aligned schedule starts a window
every $E$ blocks. We add $\varphi-1$ other starting positions within each
period, where $1\le\varphi\le e$. Specifically, choose offsets
\[
 r_j=\left\lfloor\frac{je}{\varphi}\right\rfloor,
 \qquad j=0,\ldots,\varphi-1,
\]
and start windows at $H_0+tE+r_jE/e$, for $t\ge0$. Thus $\varphi=1$
recovers alignment and $\varphi=e$ starts a window every $E/e$ blocks
(full striding). Every window is the union of $e$ aligned windows from the
level below, so the same child maps can be reused. Phase zero is always
present, preserving every aligned staircase proof.

\begin{lemma}[Staggered SuperBUDs]
\label{lem:staggered-schedule}
The rate is
$\varphi(1-e^{-L})/(e-1)$ SuperBUDs per block. Each write contributes at most
$\varphi L$ SuperBUD entries, in addition to its BUD record. Write
$q=\lceil e/\varphi\rceil$. If $\varphi\ge2$ and
\[
 1\le d\le \frac{e-q}{e}\eta,
\]
then every anchor position admits a scheduled window containing $[a,h]$
whose size, and hence the wait for its last block, is less than
$\tfrac{e^2}{e-q}d$. By that block, the anchor BUD can be paired with either this SuperBUD
(S4) or a successor BUD (S3), once the corresponding evidence is retrievable.
In the prior-value variant, the latter case needs only the successor BUD
(S3-PV); the worst-case guarantee remains at most two digests.
\end{lemma}

\subsection{Sliding SuperBUDs}
Build a SuperBUD of each size $e,\ldots,e^L$ at every block, once enough
blocks have been logged. A window can advance one block without recovering
its full write history:
\[
 M_{[s+1,t]}=\{\,k\mapsto u\in M_{[s,t]}:u>s\,\}.
\]
An entry disappears only if its last modification was at the outgoing
height $s$. Apply the next block's writes to obtain $M_{[s+1,t+1]}$, then
build its Merkle tree. This gives the most choices for a proof, but requires
rebuilding a SuperBUD of every configured size after every block.

\begin{lemma}[Sliding SuperBUDs]
\label{lem:sliding-schedule}
The rate is $L$ SuperBUDs per block. Worst-case archive
growth is $\Theta(\eta)$ SuperBUD entries per write. If each block has at
most $w$ writes, rebuilding the SuperBUD maps and Merkle trees costs $O(w\eta)$ per block,
which is attained when each block has $w$ fresh keys. Once
$h\ge H_0+\eta-1$, every $1\le d\le\eta$ has a certificate with two BUDs or SuperBUDs,
both summarizing blocks at or below $h$.
\end{lemma}

The proofs are in \cref{app:schedule-proof}.

\paragraph{Choosing a schedule.}
The aligned hierarchy is a useful default: it produces fewer than
$1/(e-1)$ SuperBUDs per block and answers every anchored query with
BUDs and SuperBUDs ending by the query height. One fixed size reduces that rate further when $m$
is large, but its proof length eventually grows as $d/m$. Staggering adds windows that can pair with an anchor BUD to give S4. Within the distance range in
\cref{lem:staggered-schedule}, two starting positions per size are the least
expensive choice in this family that guarantees such a window for every
anchor position; they produce fewer than one SuperBUD per block for $e\ge3$.
Sliding removes the need to wait for a later block throughout $d\le\eta$
after startup, at the much higher worst-case construction cost above.
\Cref{sec:eval:q5,tab:eval:sched} compares alignment, full striding, and the
fixed-size SuperBUD schedule. Sliding is analyzed here but not benchmarked.

\section{Evaluation details}
\label{app:eval:q3}
This appendix gives the benchmark configurations and detailed results
for certificates, SuperBUD schedules, and garbage collection (\cref{sec:eval}).

\subsection{Systems, workloads, and sampling}
\label{app:eval:q3:setup}

All backends implement one replay interface. BUD applies writes to
flat state and builds each BUD; its timing includes both operations. The
in-memory MPT omits RLP and Ethereum's separate account and
storage tries, reducing the work compared with a full client; the disk version is
content addressed and capped at a $16$\,MiB cache. NOMT and QMDB retain their
defaults in the short commitment sweeps. QMDB garbage collection remains
dormant there. The paired experiments in \cref{sec:eval:q7} lower its
live-entry threshold; separately disclosed storage and diagnostic runs also
change the entry-file segment size or refresh a cached size in the compactor.

The prototype implements flat-state replay, previous-write pointers,
BUD and SuperBUD construction, and hash-path certificate checks. Initial-state
activation, prior-value records, the specified digest domain separation and
canonical tree encoding,
and end-to-end attestation verification are not implemented. Certificate checks
trust the supplied BUD and SuperBUD digests; the BLS benchmark below measures signatures
separately. Thus the measurements cover construction and hash proofs under
the benchmark encoding, not an integrated implementation of the full protocol.

Each commitment-sweep point prepopulates $N$ keys and times $128$ complete
block replays, including state application and commitment construction but excluding
attestation, signing, and certificate generation. We report p50 with p99
whiskers. Runs use an Apple M4 with $16$\,GB RAM. The certificate and SuperBUD experiments run for $6\eta$ blocks, so every level produces SuperBUDs repeatedly; their
verification timings cover warm hash paths. In a separate BLS12-381 benchmark,
a client using preaggregated keys for $67$ of $100$ authenticated signers
spent about $0.65$\,ms on the two batch-verified aggregate signatures for S3 or
S4. Committee authentication and signer-key aggregation are outside that
timed operation.

The synthetic certificate workload uses $N=10^5$ keys and $w=50$ distinct
writes per block, sampled from a Zipf distribution with $s=1$. It runs for
$6\eta=98{,}304$ blocks with $\eta=16{,}384$, $e=4$, $L=7$, and a $2\%$
deletion rate. For every sampled anchor and staleness $d$, we build every
applicable strategy rather than only the cheapest.

Event sampling draws a replay event, either a synthetic modification or an
account occurrence in the Ethereum trace, and therefore weights keys by their
observed frequency. Uniform-key sampling first draws uniformly from keys
observed in the workload and then selects one of that key's events. The latter
gives rarely recurring keys equal weight without assuming that client queries
follow the event distribution.

\subsubsection{Ethereum account-access trace.}
\label{app:eval:trace}

We replay $24{,}576$ blocks spanning approximately $3.4$ days, from
height $25{,}800{,}000$ through $25{,}824{,}575$. It uses $\eta=4{,}096$, $e=4$, and $L=6$.
From public RPC data we mark transaction senders and recipients, fee and
withdrawal recipients, created contracts, and log-emitting contracts, and
replay each distinct marked account as one putative account write. The trace
contains an average of $445$ distinct marked accounts per block.

A marked account need not have changed, and an account changed through
internal execution may be missing. The RPC data also omits storage-slot
writes, account deletions, and post-execution values. It does not supply the
execution-complete block-level access list specified by
EIP-7928~\cite{eip7928}. The results measure the observed account sequence's
effect on path depth, recurrence gaps, and proof availability; they do not
estimate full-EVM write volume or client cost. Only the synthetic workload
exercises tombstone anchors and S2.

Each address receives a distinct dense identifier for replay indexing. This
preserves per-block distinct-account counts and recurrence gaps, hence proof
availability and the depths of our balanced Merkle trees. It does not preserve
native address serialization, production Patricia-trie shape, or database
locality. Byte totals use fixed-width benchmark records, not native Ethereum
account and storage keys. Replacing these with native-width keys changes the
opened-leaf payload, although not the number of hashes in a balanced-tree path.

\subsection{Certificate costs}
\label{app:eval:q3:costs}

The synthetic workload's S1 and S2 payloads cost $219$--$227$ bytes; S3
costs $452$--$454$ bytes. S4 grows from $515$ to $790$ bytes as its
SuperBUDs contain more keys, and S5 from $901$ to $11{,}984$ bytes as it
needs more BUD or SuperBUD proofs. These totals include records, entries, identifiers, and
Merkle paths, but exclude attestations and committee evidence
(\cref{sec:limits}). Every S5 exclusion count satisfies
\cref{eq:staircase}. S3 and S4 remain smaller than the prototype trie's
$2{,}087$-byte current-value inclusion proof under the same record encoding;
S5 exceeds it at larger staleness.

On the Ethereum trace, S3 costs $677$--$710$ bytes: more marked accounts
per block mean deeper BUD paths, consistent with the $O(\log w_n)$ anchor
cost in \cref{sec:costs}. Yet at $d=64$, both workloads need about $5.5$
BUD or SuperBUD proofs in the S5 cover under either sampling rule (\cref{tab:eval:crosswork}).

\subsection{Strategy availability}
\label{app:eval:q3:strategies}

\Cref{fig:eval:mix,fig:eval:mixeth} report the cheapest applicable strategy.
These shares differ from S4 availability in \cref{fig:eval:q3cost}(b), which
counts every query for which S4 exists, even when another strategy is cheaper.

On the synthetic workload, S3 serves most queries at $d=1$. S4 is more useful
for keys without a nearby successor; its availability falls with staleness
from $90.6\%$ to $32.9\%$ under event sampling and from $99.9\%$ to
$42.9\%$ under uniform-key sampling.

The Ethereum trace contains a much larger population without observed
successors. Repeated appearances have a median gap of $8$ blocks, yet $50.4\%$
of its $1.68$ million marked accounts appear only once. That fraction describes
the $3.4$-day observation window, not a steady state; a longer trace may reveal
later writes.

Under uniform-key sampling on the Ethereum trace, S4 is the most common strategy at every
sampled $d<\eta$ and retains a majority through $d=2{,}048$. At $d\ge\eta$,
S4 availability reaches zero because no configured level spans the interval.
This does not rule out a proof: the anchor at $a$ remains valid
(\cref{sec:security}), and S5 remains available from archived BUDs and SuperBUDs.

\begin{table}[t]
\caption{Workload comparison at matched staleness. Strategy shares name the
cheapest applicable strategy; S4 availability counts S4 whenever it exists.
All listed stalenesses lie below both retention windows.}
\label{tab:eval:crosswork}
\centering
\small
\setlength{\tabcolsep}{5pt}
\begin{tabular}{lcc@{\hskip 14pt}cc}
\toprule
& \multicolumn{2}{c}{Write-event sampling}
& \multicolumn{2}{c}{Uniform-key sampling} \\
\cmidrule(r){2-3}\cmidrule(l){4-5}
& synthetic & Ethereum & synthetic & Ethereum \\
\midrule
S3 share at $d=1$           & 97.0\% & 79.0\% & 84.2\% & 28.4\% \\
S4 share at $d=1$           & 3.0\%  & 21.0\% & 15.8\% & 71.5\% \\
S4 availability at $d=512$ & 69.2\% & 82.3\% & 84.6\% & 88.5\% \\
S5 cover size at $d=64$       & 5.55   & 5.52   & 5.47   & 5.61   \\
\bottomrule
\end{tabular}
\end{table}

\begin{figure}[t]
\centering
\includegraphics[width=0.82\linewidth]{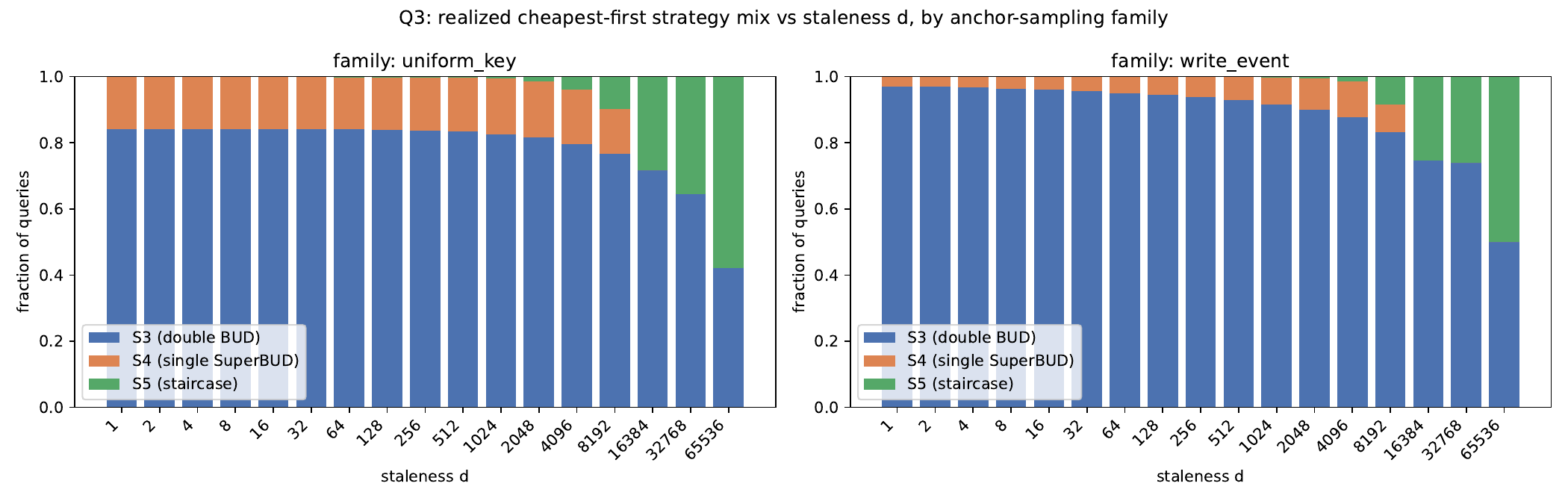}
\caption{Cheapest applicable strategy by staleness on the synthetic workload
($\eta=16{,}384$, $e=4$, $L=7$, $w=50$), under uniform-key sampling (left)
and event sampling (right).}
\label{fig:eval:mix}
\end{figure}

\begin{figure}[t]
\centering
\includegraphics[width=0.82\linewidth]{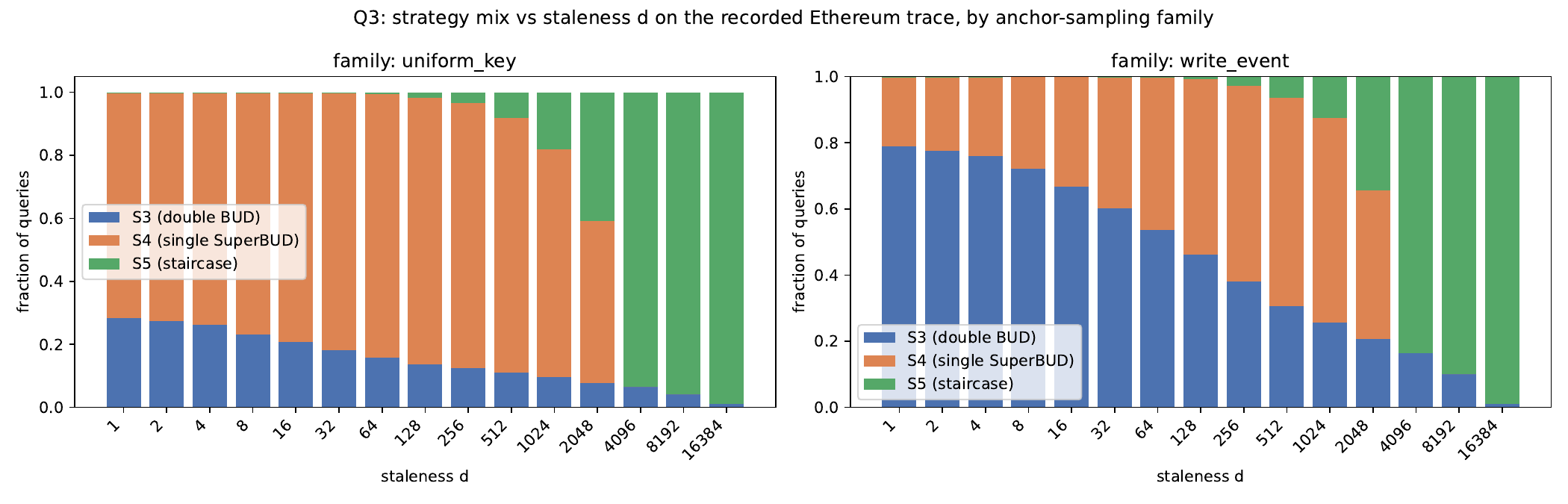}
\caption{Cheapest applicable strategy by staleness on the Ethereum account-access
trace
($\eta=4{,}096$, $e=4$, $L=6$, $w\approx445$), under uniform-key sampling
(left) and event sampling (right). Accounts observed only once favor S4
over S3 under uniform-key sampling.}
\label{fig:eval:mixeth}
\end{figure}

\subsection{Hierarchy parameters}
\label{app:eval:params}

The SuperBUD hierarchy results are in the main text (\cref{tab:eval:params}). The
sweep fixes $N=10^5$, $w=50$, $\eta=4{,}096$, and $d\in[64,128)$,
and runs each configuration for $6\eta$ blocks. Every sampled S5 exclusion
count satisfies \cref{eq:staircase}.

\subsection{Skew sensitivity}
\label{app:eval:skew}

We repeat the synthetic experiment with $s\in\{0.8,1,1.2\}$ at
$(e,L)=(4,6)$. Across all queries, mean S5 size changes by $35\%$, largely
because greater skew produces fresher anchors. Restricting the comparison to
$d\in[64,128)$ reduces the variation to $5.6\%$; the mean cover remains
between $5.98$ and $6.01$ BUD or SuperBUD proofs. Skew therefore changes anchor age much more
than proof cost at a fixed staleness.

\subsection{Schedule sweep}
\label{app:eval:schedules}

\Cref{tab:eval:sched} replays the same two $4{,}096$-block traces under
the aligned schedule, full striding, and a single SuperBUD size $m=64$.
The synthetic trace has $w=50$; the Ethereum segment averages
$w\approx478$. The two-phase SuperBUD schedule ($\varphi=2$) is not benchmarked;
its bounded-wait guarantee is analytical (\cref{lem:staggered-schedule}).
Per-block sliding is also treated analytically, with worst-case construction
work $\Theta(w\eta)$.

S4 availability in \cref{tab:eval:sched} is aggregated over log-uniform valid
historical queries at the head of a one-window replay and is not directly
comparable with the per-staleness values in \cref{fig:eval:q3cost}(b). For each
retained anchor, the sampler records its first later modification, rejects a
draw that crosses that successor or the replay head, and verifies the complete
staircase certificate. Wait medians include only queries for which S4 appears
before the replay ends; queries still waiting are excluded. These conditions
apply equally to all three schedules.

\begin{table}[t]
\caption{Schedule comparison over one $4{,}096$-block replay, at
$\eta=4{,}096$ and $e=4$, with the single level sized
$m=64=\sqrt{\eta}$. Within each workload, all schedules replay the same trace.
SuperBUD counts exclude the mandatory BUD. S4 availability is measured at query
time. Wait is in blocks, and its median includes only queries for which S4
becomes available before the replay ends.}
\label{tab:eval:sched}
\centering
\small
\setlength{\tabcolsep}{4pt}
\begin{tabular}{lccccc}
\toprule
Schedule & \shortstack{SuperBUDs\\per block}
& \shortstack{Archive\\B/blk}
& \shortstack{S5 cover size\\p50/p99}
& \shortstack{S4\\avail.}
& \shortstack{Wait p50\\(blocks)} \\
\midrule
\multicolumn{6}{l}{\emph{Synthetic Zipf trace} ($w=50$)} \\
$\varphi=1$ (aligned) & 0.3333 & 4{,}252  & 4/12  & 83.7\% & 14 \\
$\varphi=e$ (strided) & 1.3286 & 11{,}810 & 4/12  & 96.0\% & 2  \\
Single level          & 0.0156 & 1{,}959  & 13/71 & 46.1\% & 23 \\
\addlinespace
\multicolumn{6}{l}{\emph{Ethereum trace, first $\eta$ blocks}
($w\approx478$)} \\
$\varphi=1$ (aligned) & 0.3333 & 41{,}904  & 4/12  & 83.3\% & 14 \\
$\varphi=e$ (strided) & 1.3286 & 110{,}813 & 4/12  & 96.6\% & 2  \\
Single level          & 0.0156 & 20{,}215  & 10/72 & 39.6\% & 22 \\
\bottomrule
\end{tabular}
\end{table}

\subsection{QMDB garbage collection: work, stalls, and storage}
\label{sec:eval:q7}

The short sweeps do not activate QMDB's garbage collector. The experiments
below isolate what changes when it runs: live-entry relocation consumes CPU
and device writes, and a separate scheduling issue in the tested version can
delay block completion. We also distinguish bytes appended from storage
retained after old files are deleted.

\begin{figure}[t]
\centering
\includegraphics[width=\linewidth]{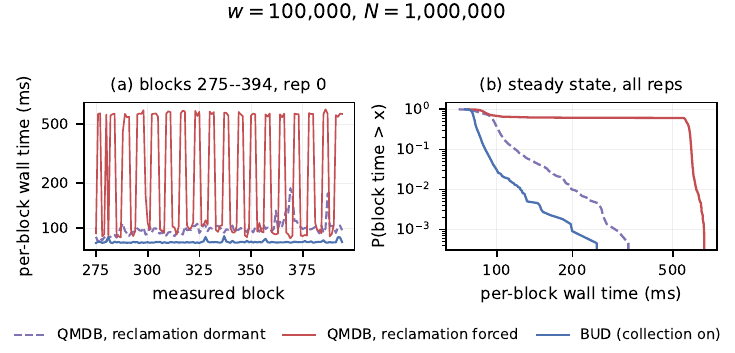}
\caption{Block-completion delays with QMDB garbage collection active at
$N=10^6$, $w=10^5$, and $2\%$ deletions. The active configuration lowers
the live-entry threshold to $5\cdot10^4$ per shard; the control keeps
the default threshold, at which collection is dormant at this state size.
(a) Blocks $275$--$394$ of repetition~0 show recurring delays of roughly
half a second. (b) Empirical complementary CDF pooled across both
repetitions of the $1{,}100$-block replay. The BUD curve is contextual:
it excludes persistent storage, SuperBUD maintenance, attestation, and
signing, whereas QMDB performs durable writes.}
\label{fig:eval:q7}
\end{figure}

\paragraph{Configuration and mechanism.}
We test QMDB v0.2.0, commit \texttt{f14a2a0}. Its garbage collector runs
when a shard exceeds \texttt{compact\_thres} live entries and utilization,
the fraction of log entries that remain live, falls below $0.7$. The default
gate is $2{\cdot}10^7$ live entries per shard, approximately
$3.2{\cdot}10^8$ keys across $16$ evenly populated shards. Our states are
smaller, so the default configuration leaves garbage collection inactive.
To exercise the same relocation code, we lower the gate to
$5{\cdot}10^4$ per shard at $N=10^6$ and $5{\cdot}10^5$ at $N=10^7$.
The utilization rule and relocation algorithm remain unchanged. Creates
append both a new entry and its rewritten predecessor, so prepopulation by
creates starts at approximately $50\%$ utilization. Garbage collection
activates during prepopulation once the lowered live-entry threshold is crossed.

Once active, garbage collection can relocate up to two live entries per
create, write, or delete. The shard updater appends each relocated entry
under a new serial number and updates the index and cache; the flusher then
appends its authenticated leaf and deactivates the old one. Block completion
waits for the updaters and the flusher. The latter performs one file
synchronization per shard at the end of each block. Garbage collection thus
adds authenticated update work to the block-completion path, even when
parallelism or other bottlenecks hide that work in wall time.

\paragraph{Measuring work.}
Each active/inactive pair replays a byte-identical trace with $2\%$
deletions. We measure process CPU time across all threads, kernel-reported
device bytes written, and entries appended from the summed serial-number
counters. Subtracting the appends required by the logical operations---two
per create and one per write or delete---gives the live-entry relocation
count. This difference is zero in every measured block of every inactive
run. Unless indicated otherwise in \cref{tab:eval:q7}, each point has two
repetitions; they agree within $1\%$ on relocations, appends, and device
bytes, and within $7\%$ on CPU time.

\begin{table}[t]
\caption{QMDB work with garbage collection active versus dormant on paired
traces with $2\%$ deletions. CPU and device-byte factors are per logical
write; relocations and appends are per QMDB operation. Dormant appends are
approximately $1.02$ per operation. Active runs lower the live-entry gate;
two repetitions per point except $\dagger$, which has one.}
\label{tab:eval:q7}
\centering
\footnotesize
\setlength{\tabcolsep}{5pt}
\begin{tabular}{@{}rrcccc@{}}
\toprule
$N$ & $w$ & \shortstack{Relocations\\per operation}
& \shortstack{Appends\\per operation}
& \shortstack{CPU\\factor}
& \shortstack{Device bytes\\factor} \\
\midrule
$10^6$ & $2{,}000$   & 1.60 & 2.62 & $1.8\times$ & $1.6\times$ \\
$10^6$ & $20{,}000$  & 1.39 & 2.41 & $2.6\times$ & $2.2\times$ \\
$10^6$ & $100{,}000$ & 1.13 & 2.15 & $2.7\times$ & $2.1\times$ \\
$10^7{}^{\dagger}$ & $2{,}000$   & 1.64 & 2.65 & $1.3\times$ & $1.6\times$ \\
$10^7{}^{\dagger}$ & $100{,}000$ & 1.45 & 2.47 & $1.9\times$ & $2.4\times$ \\
\bottomrule
\end{tabular}
\end{table}

At $N=10^6$, active garbage collection relocates $1.13$--$1.60$ live entries
per operation. Appends rise from $1.02$ to $2.15$--$2.62$ per operation, CPU
time per logical write rises by $1.8$--$2.7\times$, and device bytes written
rise by $1.6$--$2.2\times$. At $N=10^7$, CPU increases are
$1.3\times$ and $1.9\times$ for $w=2{,}000$ and $10^5$, with device-write
increases of $1.6\times$ and $2.4\times$. These larger-state points have
one repetition each.

Wall time need not increase with this extra work. QMDB uses $16$ updater
threads, $16$ compactor threads, and a flusher on our $10$-core machine.
In the earlier $N=10^7$, $w=2{,}000$ run, the $16$ synchronous file
flushes accounted for about $31$\,ms per block; predecessor reads also
contributed substantially.
At the latter state size and $w=10^5$ (\cref{fig:eval:q7large}), the active
configuration has a median block time only $0.55\times$ that of the inactive configuration, despite
using more CPU and writing more device bytes. Improved locality after
relocation is a possible explanation, but these measurements do not isolate
its contribution. The work measurements support a garbage-collection cost
even where a median-latency comparison does not.

\begin{table}[t]
\caption{Absolute CPU costs and block-latency outcomes for
\cref{tab:eval:q7}. CPU is in $\mu\mathrm{s}$ per logical write.
Median factors include stalled blocks. A block is counted as stalled if it
exceeds the dormant run's rolling median by more than $300$\,ms.
The dense $N=10^6$ row summarizes the full $1{,}100$-block replay, not only
its final $600$ blocks. $\dagger$ denotes one repetition.}
\label{tab:eval:q7latency}
\centering
\footnotesize
\setlength{\tabcolsep}{4.5pt}
\begin{tabular}{@{}rrccc@{}}
\toprule
$N$ & $w$ & \shortstack{CPU $\mu$s/write\\dormant $\to$ active}
& \shortstack{Median block\\active/dormant}
& \shortstack{Stalled\\blocks} \\
\midrule
$10^6$ & $2{,}000$ & $10.84\to19.84$ & $1.16\times$ & $3\%$ \\
$10^6$ & $20{,}000$ & $2.04\to5.21$ & $1.23\times$ & $21\%$ \\
$10^6$ & $100{,}000$ & $1.33\to3.65$ & $5.80\times$ & $62\%$ \\
$10^7{}^{\dagger}$ & $2{,}000$ & $8.82\to11.88$ & $1.02\times$ & $0\%$ \\
$10^7{}^{\dagger}$ & $100{,}000$ & $1.86\to3.52$ & $0.55\times$ & $12\%$ \\
\bottomrule
\end{tabular}
\end{table}

\paragraph{Waiting on the compactor.}
The tested implementation also exhibits recurring stalls
(\cref{fig:eval:q7}). In \texttt{fill\_compact\_chan}, the compactor caches
the entry-file size and refreshes it only inside a $500$\,ms sleep. When
the scan reaches that cached size, the compactor can sleep although the file
has grown. Once the updater drains the $1{,}024$-job ring, it blocks inside
\texttt{compact()} waiting for the next job. Sampled stacks in a dense run
place the updaters in that channel wait for $82\%$ of their non-idle samples
and the compactors in their sleep for $98\%$ of theirs. Relocation itself
accounts for $5\%$ of the updater samples.

As reported in \cref{tab:eval:q7latency}, a block is stalled when its time exceeds the inactive run's
rolling median at the same block by more than $300$\,ms. Typical excesses
are $470$--$490$\,ms. At $N=10^6$, stalls affect $3\%$ of blocks in the
$5{,}000$-block run with $w=2{,}000$, $21\%$ in the $1{,}000$-block run
with $w=20{,}000$, and $62\%$ across the full $1{,}100$-block run with
$w=10^5$. In the dense run they become persistent: over blocks
$500$--$1{,}099$, median block time is approximately $572$\,ms active
versus $100$\,ms inactive ($5.7\times$), with repetition medians
$571/573$\,ms and $98/101$\,ms respectively. This late interval is distinct from the
earlier time-series excerpt in \cref{fig:eval:q7}(a).

A diagnostic build refreshes the cached size on every scan iteration. Its
$300$-block replay contains $0.7\%$ stalled blocks, versus $33$--$38\%$
in the original $400$-block runs on the same trace; the corresponding p90
block times are about $100$ and $590$\,ms. CPU cost remains similar
($3.40$ versus $3.56\,\mu\mathrm{s}$ per write), while appends,
relocations, and device bytes are unchanged. The diagnostic therefore
separates a removable waiting cost from the relocation work. We report the unmodified relocation code in the main results; the
diagnostic does not establish that the observed stalls are inherent to
QMDB's design. Their occurrence at the default activation scale is unmeasured.

\begin{figure}[!htbp]
\centering
\includegraphics[width=0.96\linewidth]{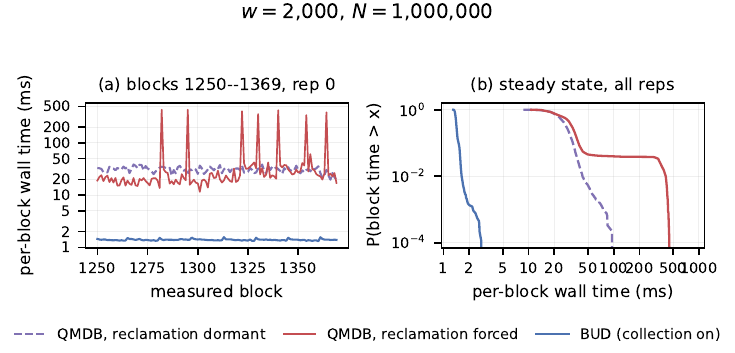}
\par\smallskip
\includegraphics[width=0.96\linewidth]{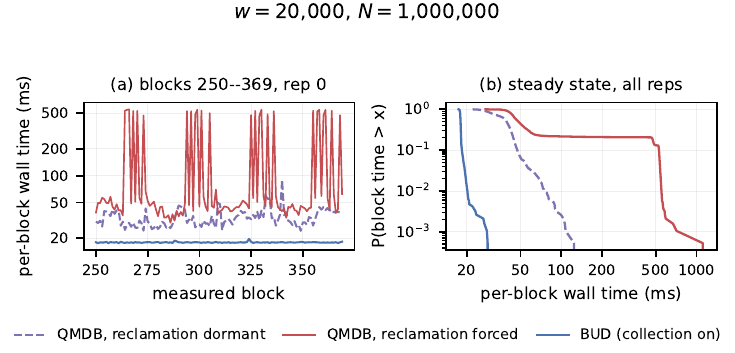}
\caption{Smaller blocks at $N=10^6$: $w=2{,}000$ (top) and
$w=20{,}000$ (bottom). Each plot shows a $120$-block excerpt on the left
and the pooled empirical latency tail on the right. The reported stall
fractions over the full runs are $3\%$ and $21\%$, respectively. Rare
half-second delays can be almost invisible in the median. Both points
have two repetitions and $2\%$ deletions; the BUD durability
qualification in \cref{fig:eval:q7} also applies here.}
\label{fig:eval:q7smaller}
\end{figure}

\begin{figure}[!htbp]
\centering
\includegraphics[width=\linewidth]{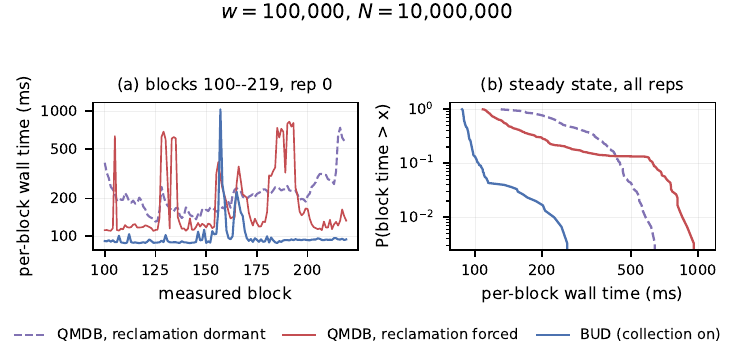}
\caption{Dense replay at $N=10^7$ and $w=10^5$, with $2\%$ deletions
and one repetition. Active garbage collection has a lower median than the
dormant control but a heavier upper tail. The work measurements still show
$1.9\times$ process CPU time and $2.4\times$ device bytes written
(\cref{tab:eval:q7}). The active live-entry threshold is
$5\cdot10^5$ per shard. The BUD durability qualification in
\cref{fig:eval:q7} also applies here.}
\label{fig:eval:q7large}
\end{figure}

\paragraph{Appended bytes and retained storage.}
The earlier experiment's $2.6\times$ file-growth difference did not measure
storage after garbage collection had an opportunity to delete complete
files. QMDB unlinks whole entry-file segments when its periodic prune,
every $500$ heights, passes them. Its default segment size is $1$\,GiB per
shard, whereas no per-shard file in that experiment exceeded
$0.54$\,GiB. Consequently, no complete segment could be released.

In the dense $N=10^6$ storage experiment, both configurations use
$64$\,MiB segments, with the disclosed gate change in the active
configuration. At each of the pruning heights $510$ and $1{,}010$, the
active configuration drops from $131$ to $19$ retained files. Across the $1{,}100$-block run, retained files range from $0.5$ to
$7.8$\,GB despite $16.3$\,GB being appended. The inactive configuration
never shrinks and retains its $7.7$\,GB of appended data. These finite runs
show that relocation enables space to be recovered; they do not establish a
general storage bound. They also explain why additional bytes appended
cannot be equated with additional storage retained indefinitely. In the
inactive configuration, median block time rises from $75$ to $123$\,ms
over the run, while device reads per write rise from $0.1$ to $48$ bytes.

\paragraph{What BUDs avoid.}
The BUD authenticates the current block's write log and has
no corresponding live-entry relocation. Expired tombstones can be deleted
without rewriting other authenticated entries. In the $5{,}000$-block
$w=2{,}000$ run, BUD collects $1.1$ tombstones per block after collection
starts, using $0.3\%$ of block time; median block time before and after
that height differs by $0.6\%$ and $1.7\%$ in the two repetitions.
Across the three write volumes, BUD cost is
$0.71$--$0.91\,\mu\mathrm{s}$ per write. These figures exclude SuperBUD
maintenance, attestation, signing, and a persistent key-value engine. The
replay performs no durable writes, so we do not turn its timings into a
speedup over QMDB. A persistent implementation still pays its storage
engine's maintenance costs; the architectural difference is that such
maintenance does not require rebuilding BUDs.

\begin{samepage}
\paragraph{Scope.}
All active-garbage-collection results lower the live-entry gate. The storage experiment additionally uses $64$\,MiB segments, and the
diagnostic changes the compactor's cached-size refresh. We have not measured garbage collection near the
default gate or CPU cost per relocation beyond $N=10^7$. Symmetric
prepopulation makes shards reach their cached scan limits together and may
affect the stall fraction. The experiments drive one block at a time on
macOS without \texttt{io\_uring}, use synthetic Zipf workloads with one-byte
values, and keep the BUD state in memory. The results establish extra
CPU and device-write work under active QMDB garbage collection and expose
one implementation-specific way that maintenance can delay block
completion. They do not claim an end-to-end throughput advantage under
matched durability.
\par
\end{samepage}

\clearpage
\section{Correctness proofs}
\label{app:correctness}
\label{sec:correctness}

We work under the assumptions of \cref{sec:model,sec:attest} and use the
verification rules of \cref{sec:cert}.

\subsection{Predecessor records}

\begin{lemma}[Retained metadata]
\label{lem:retained-metadata}
After block $n\ge H_0$, a present entry for $k$ stores
$\last(k)=\max(\Mod(k)\cap[H_0,n])$ if that set is nonempty. Otherwise,
$k\in\operatorname{dom}(\sigma_0)$ and $\last(k)=H_0-1$.
\end{lemma}

\begin{proof}
Initialization assigns $H_0-1$ to inherited entries. Every write at height
$m$ stores $\last(k)=m$, replacing any earlier metadata. Garbage collection
only removes entries; it never changes the height of a retained entry.
\end{proof}

\begin{lemma}[Missing entries]
\label{lem:missing-entry}
If $k$ has no local entry after block $n\ge H_0$, then
$\val(k,j)=\absent$ for every $j\in[n-\eta,n]\cap[H_0,n]$.
With unbounded retention, this holds throughout $[H_0,n]$.
\end{lemma}

\begin{proof}
If $k$ has no logged modification, it was absent initially: an unwritten
inherited entry would still be present. Hence it has remained absent.

Otherwise, let $m$ be its last modification through $n$. A live value
written at $m$ could not have been collected, so this write was a deletion.
Its tombstone must have been collected at some height $c\le n$, requiring
$m+\eta\le c$. Thus $m\le n-\eta$, and no write follows $m$ through $n$.
The key is absent throughout $[m,n]$, which contains the claimed interval.
With unbounded retention, no tombstone is collected, so only the first
case is possible.
\end{proof}

\begin{lemma}[Predecessor records]
\label{lem:chain}
\label{lem:chain:full}
Every canonical record $(k,v,b,p)$ satisfies:
\begin{enumerate}[label=(\roman*),leftmargin=2em,itemsep=1pt]
\item If $p\ge H_0$, then $p=\max(\Mod(k)\cap[H_0,b))$.
\item If $p=-1$, then $\val(k,h)=\absent$ for every
$h\in[b-\eta,b)\cap[H_0,b)$.
\item If $p=H_0-1$, then $k\in\operatorname{dom}(\sigma_0)$ and
$\Mod(k)\cap[H_0,b)=\emptyset$.
\end{enumerate}
\end{lemma}

\begin{proof}
The record copies the metadata present before the block's first write to
$k$. For $b>H_0$, \cref{lem:retained-metadata} at height $b-1$ gives
(i) and (iii), while \cref{lem:missing-entry} gives (ii). At $b=H_0$,
(i) cannot occur, the interval in (ii) is empty, and (iii) follows from
initialization. Further writes in the same block preserve this predecessor.
\end{proof}

In the baseline format, case (iii) identifies an inherited entry but does
not authenticate its value before $b$. The prior-value format additionally
commits that value.

\begin{lemma}[Prior-value records]
\label{lem:prior-value}
Each canonical record $(k,v^-,v,b,p)$ committed by $U_b^+$
satisfies $\val(k,b)=\out(v)$. If $p\ge H_0$, then
$\val(k,h)=\out(v^-)$ for $p\le h<b$. If $p=H_0-1$, then
$\val(k,h)=\out(v^-)=\sigma_0(k)$ for $H_0\le h<b$.
If $p=-1$, then $v^-=\tomb$ and the bounded absence statement in
\cref{lem:chain}(ii) holds.
\end{lemma}

\begin{proof}
The current field is the final logged value at $b$, so \cref{eq:val}
gives the result at that height. The prior-value and predecessor fields
are copied together before the first write in the block. For $p\ge H_0$,
\cref{lem:chain}(i) rules out every modification in $(p,b)$, so the copied
value is the value written at $p$, including a retained deletion marker.
For $p=H_0-1$, \cref{lem:chain}(iii) says that the entry is inherited
and no logged write precedes $b$; its copied value is therefore
$\sigma_0(k)$. For $p=-1$, no local entry exists, so the construction
copies $\tomb$, and \cref{lem:chain}(ii) applies without change.
\end{proof}

\subsection{Authenticated maps}

\begin{lemma}[Authenticated-map binding and exclusion]
\label{lem:merkle-map}
\label{lem:merkle-map:full}
Let $D=\MR_{T,I}(\mathbf z)$ commit to a canonical, strictly key-sorted
sequence $\mathbf z$. Except with negligible probability, an accepting
inclusion proof at position $i$ opens $z_i$, and an accepting exclusion
proof for $k$ implies that $\mathbf z$ contains no entry with key $k$.
\end{lemma}

\begin{proof}
Compare the verifier's computation with the canonical tree from
\cref{sec:maps}. The outer hash binds the object type, index, entry count,
and tree root. Different values for any of these fields under the same
digest give a collision in $H$. The domain separators likewise distinguish
empty and nonempty maps.

With the canonical count and root fixed, the checked position determines
the leaf, path length, and left--right choices. If the opened entry differs
from $z_i$, injectivity of the encoding gives distinct leaf inputs. Either
their hashes agree, yielding a collision, or their hashes differ. In the
latter case, follow the supplied and canonical paths upwards. At the first
level where the hashes agree, distinct ordered pairs of child hashes
produce the same parent hash. This is again a collision. For a one-leaf
tree, the leaf hash is already the root.

An interior exclusion proof therefore opens the canonical entries at two
consecutive positions. Their keys straddle $k$, leaving no position for it
in the sorted sequence. At a boundary, the proof opens the first entry to
a key greater than $k$, or the last entry to a key less than $k$; sorting
again rules out $k$. The empty case follows from its distinct encoding.
A false accepted proof thus yields a hash collision, which a polynomial-time
adversary finds only with negligible probability.
\end{proof}

\subsection{Soundness}
\label{sec:security}
\label{app:security-full}

\begin{theorem}[Soundness]
\label{thm:soundness}
\label{thm:soundness:full}
The read system is sound for both the baseline and prior-value formats.
\end{theorem}

\begin{proof}
Every accepted BUD or SuperBUD has signatures from more than $f$ distinct members
of its authenticated committee, including an honest signer. Unless a
signature was forged, that signer attested the exact message being checked
and hence, by the attestation rule, the canonical digest. By
\cref{lem:merkle-map}, the openings are also canonical unless a hash
collision occurs. A union bound over the polynomially many signatures and
openings gives negligible failure probability.

Condition on neither failure occurring. An ordinary anchor opens a record
$(k,v,a,p)$ in $U_a$, with $H_0\le a\le h$. Block $a$ therefore wrote $v$
to $k$. Each accepted extension rules out a later modification through $h$:
\begin{itemize}[leftmargin=1.5em,itemsep=2pt]
\item \textsc{Empty} requires $a=h$.
\item \textsc{Next} opens a record $(k,v',b,a)$ with $b>h$.
By \cref{lem:chain}(i), $a$ is the last modification before $b$,
so none occurs in $(a,h]$.
\item \textsc{Cover} excludes $k$ from BUDs or SuperBUDs, proving that
their intervals contain no modification. An opening of $k\mapsto a$
in a SuperBUD instead rules out modifications after $a$ within its window.
The verifier checks that these intervals together cover every height in
$(a,h]$.
\end{itemize}
Thus $a$ is the last modification through $h$. By \cref{eq:val},
$\val(k,h)=\out(v)$, which is the required output. This also covers
deletions, since $\out(\tomb)=\absent$.

A sentinel anchor opens $(k,v,b,-1)$ and requires
$H_0\le h<b$ and $b-\eta\le h$. By \cref{lem:chain}(ii),
$\val(k,h)=\absent$, again the required output. The verifier does not
accept $H_0-1\ge0$ in place of $-1$.

For enriched BUDs, the distinct type and protocol version authenticate the
enriched format, and binding fixes the prior value as well as the baseline
fields. Projecting to the current-value and predecessor fields preserves
all the preceding arguments for ordinary and sentinel anchors. For a
single-BUD certificate from \cref{sec:prior-cert}, \cref{lem:prior-value}
gives precisely the required output in each of its four accepted cases.
In particular, $h=b$ uses the current value, and $p=-1$ does not permit a
query below the sentinel cutoff. These exhaust the accepted forms in both
formats, so an incorrect output is accepted only with negligible probability.
\end{proof}

\subsection{Completeness}

Completeness concerns the query set $Q$ defined in \cref{sec:cert}:
keys with a logged modification at or before the query height. As in
\cref{sec:model}, the archived BUD and SuperBUD data, valid attestations, and
authenticated committee evidence needed for the certificate must be
retrievable.

\begin{theorem}[Completeness]
\label{thm:completeness}
\label{thm:completeness:full}
The read system is complete for both the baseline and prior-value formats.
\end{theorem}

\begin{proof}
Take $(k,h)\in Q$ and let $a=\max(\Mod(k)\cap[H_0,h])$. Open the
record $(k,v,a,p)$ in $U_a$ as an ordinary anchor and return $y=\out(v)$.
If $a=h$, \textsc{Empty} completes the certificate. Otherwise, each BUD
$U_n$ with $a<n\le h$ excludes $k$, since no modification follows $a$
through $h$. These exclusions form an accepting \textsc{Cover}.

The certificate uses finitely many canonical BUDs, all at heights at
most $h$. Their openings, valid attestations, and authenticated committee
evidence satisfy every verification check. The construction works for every
SuperBUD schedule and needs no successor write or touch transaction.

The same construction works with enriched BUD openings: their current
values and predecessor links are unchanged, and their exclusions clear
exactly the same blocks. It needs no successor. When an enriched successor
is available, \cref{sec:prior-cert} may instead give a shorter certificate;
an inherited-entry successor also covers its stated queries outside $Q$.
Neither shortcut is needed for the completeness claim.
\end{proof}

The aligned hierarchy gives a shorter certificate.

\begin{lemma}[Aligned certificate size]
\label{lem:aligned-certificate}
For $(k,h)\in Q$ with last modification $a<h$, the aligned hierarchy
admits a certificate using at most $B(h-a)+2$ distinct BUDs or SuperBUDs, all
summarizing blocks at or below $h$.
\end{lemma}

\begin{proof}
Put $d=h-a$ and choose
\[
\ell'=\min\{\lfloor\log_e d\rfloor,L\},\qquad q=e^{\ell'}.
\]
Let $W^\star$ be the aligned level-$\ell'$ window containing $a$, and
let $t$ be its final height. Since $0\le t-a\le q-1<d$, we have $t<h$.
Its map contains $k\mapsto a$, which clears $(a,t]$. At level zero,
$t=a$ and the BUD anchor itself suffices.

The remainder $[t+1,h]$ has length $R=h-t\le d$ and starts at a
level-$\ell'$ boundary. Take as many whole windows of that level as fit,
then repeat at each lower level down to zero. A boundary at one level is
also a boundary at the next, so this partitions the remainder exactly.
All windows end by $h$ and exclude $k$.

The cover uses $\lfloor R/e^{\ell'}\rfloor$ windows at the top level
and at most $e-1$ at each lower level. If $d<e^{L+1}$, then
$\ell'=\lfloor\log_e d\rfloor$ and $R<e^{\ell'+1}$, giving at most
$(e-1)(\ell'+1)$ exclusions. Otherwise $\ell'=L$, giving at most
$\lfloor d/e^L\rfloor+(e-1)L$. These are the two branches of $B(d)$
in \cref{eq:staircase}. Adding $U_a$ and, when distinct, the anchoring
SuperBUD gives the claimed bound.
\end{proof}

\paragraph{Availability and delay.}
Completeness does not require a bound on message delay. If every required
BUD or SuperBUD, its attestation, proof data, and committee evidence is
retrievable within $\lambda$ blocks of its last included height,
and remains retrievable, the certificate is available by $h+\lambda$.
This delivery guarantee is separate from soundness and certificate
existence; neither theorem makes an archive or signer respond.

\subsection{Schedule bounds}
\label{app:schedule-proof}

\begin{proof}[Proof of \cref{lem:single-size}]
Cover the blocks before the first size-$m$ boundary and after the last
complete size-$m$ window with at most $m-1$ BUDs at each end. The remaining
interval has at most $\lfloor d/m\rfloor$ complete windows. Every write
belongs to one such fixed-size window and can supply at most one entry to
its map. Balancing the terms $d/m$ and $2m$ gives the stated choice of $m$.
\end{proof}

\begin{proof}[Proof of \cref{lem:staggered-schedule}]
Each level produces $\varphi$ SuperBUDs per $e^\ell$ blocks. Summing over
$\ell=1,\ldots,L$ gives the SuperBUD rate. A write belongs to at most one
window for each starting position at each level, giving the storage bound.

The greatest gap between starts at size $E$ is $qE/e$. A window contains
$[a,h]$ if its start lies in $[h-E+1,a]$, which contains $E-d$ integer
positions. Such a start therefore exists whenever
\[
 E-d\ge qE/e,\qquad\text{or equivalently}\qquad
 E\ge\frac{e}{e-q}d.
\]
This also holds during startup: if the interval of possible starts reaches
below $H_0$, it contains the first start $H_0$. Since $q<e$ for
$\varphi\ge2$, the restriction on $d$ guarantees that some configured
size meets the inequality. The smallest such power of $e$ is less than
$\tfrac{e^2}{e-q}d$. Its interval ends fewer than $E$ blocks after $h$.

If the key is unchanged through that last block, the map still contains
$k\mapsto a$, yielding S4 with the anchor BUD. Otherwise let $b$ be its
first subsequent modification. Because $a$ is the last modification at or
below $h$, we have $b>h$. Also $b-a<E\le\eta$, so even a tombstone at
$a$ has not been collected before $b$. The record at $b$ therefore points
to $a$, yielding S3. With prior-value records, the same successor alone
yields S3-PV and leaves the wait bound unchanged.
\end{proof}

\begin{proof}[Proof of \cref{lem:sliding-schedule}]
At size $E$, each write belongs to at most $E$ sliding windows. The
sum of the configured sizes is
$\sum_{\ell=1}^Le^\ell=\Theta(e^L)=\Theta(\eta)$.
With fresh keys, each window has $wE$ entries, so rebuilding all its maps
and Merkle trees after each block costs $\Theta(w\eta)$. Repeated writes
to fewer keys can reduce both costs.

For the certificate, choose $E=e^{\lceil\log_e d\rceil}$. If $E=1$,
then $d=1$ and the BUD $U_h$ excludes $k$. Otherwise, take the size-$E$
SuperBUD ending at $h$. If $E=d$, its interval is exactly $(a,h]$, so
it excludes $k$. If $E>d$, it contains $a$ and opens $k\mapsto a$.
Together with $U_a$, each case gives the claimed certificate.
The startup condition keeps the window at or above $H_0$; earlier
queries can use BUD exclusions or the aligned staircase.
\end{proof}

\end{document}